\documentclass{article}

\usepackage{arxiv}

\usepackage[utf8]{inputenc} 
\usepackage[T1]{fontenc}    
\usepackage{hyperref}       
\usepackage{url}            
\usepackage{booktabs}       
\usepackage{amsfonts}       
\usepackage{nicefrac}       
\usepackage{microtype}      
\usepackage{lipsum}		
\usepackage{amsthm}
\usepackage{glossaries}
\usepackage{bussproofs}
\usepackage{rotating}
\usepackage{amssymb}
\usepackage{amsmath}
\usepackage{float}
\newtheorem{theorem}{Theorem}

\newtheorem{proposition}{Proposition}
\newtheorem{lemma}{Lemma}
\newtheorem{definition}{Definition}

\newtheorem{example}{Example}

\newenvironment{bprooftree}
{\leavevmode\hbox\bgroup}
{\DisplayProof\egroup}
\newacronym{zfc}{ZFC}{Zermelo-Frankel-Choice}
\newacronym{pif}{$\Pi - F$}{dependent function formation}
\newacronym{pii}{$\Pi - I$}{dependent function introduction}
\newacronym{pie}{$\Pi - E$}{dependent function elimination}
\newacronym{pic}{$\Pi - C$}{dependent function computation}
\newacronym{sif}{$\Sigma - F$}{dependent sum formation}
\newacronym{sii}{$\Sigma - I$}{dependent sum introduction}
\newacronym{sie1}{$\Sigma - E_{1}$}{first dependent sum elimination}
\newacronym{sie2}{$\Sigma - E_{2}$}{second dependent sum elimination}
\newacronym{sic}{$\Sigma - C$}{Computação da Soma Dependente}
\newacronym{n0}{$\mathbb N-I_{0}$}{zero introduction}
\newacronym{ns}{$\mathbb N-I_{s}$}{successor induction}
\newacronym{ifi}{$Id^{int} - I$}{identity type introduction}
\newacronym{iff}{$Id^{int} - F$}{identity type formation}
\newacronym{ie}{$Id - E$}{identity type elimination}
\newacronym{ii}{$Id - Eq$}{identity type equality}
\newacronym{ife}{$Id^{ext} - F$}{extensional identity type formation}
\newacronym{iie}{$Id^{ext} - I$}{extensional identity type introduction}
\newacronym{iee}{$Id^{ext} - E$}{extensional identity type elimination}
\newacronym{iei}{$Id^{ext} - Eq$}{extensional identity type equality}
\newacronym{bi}{$\lambda\beta\eta$}{Teoria da Beta-Eta Igualdade}
\newacronym{pf}{$\times - F$}{product formation}
\newacronym{pi}{$\times - I$}{product introduction}
\newacronym{pe1}{$\times - E_{1}$}{first product elimination}
\newacronym{pe2}{$\times - E_{2}$}{second product elimination}
\newacronym{rw}{$\mathit{LND_{EQ}-TRS}$}{Term Rewrite System}
\newacronym{uip}{UIP}{Uniquiness of Identity Proofs}

\title{A Topological Application of Labelled Natural Deduction}

\author{
Tiago Mendon\c{c}a Lucena de Veras\\ 
 Departamento de Matem\'atica\\
 Universidade Federal Rural de Pernambuco\\
 Recife, Brasil \\
  \texttt{tiago.veras@ufrpe.br} \\
   \And
 Arthur F. Ramos \\
  Microsoft, Redmond, USA \\
  \texttt{arfreita@microsoft.com} \\
 \And
  Ruy J. G. B. de Queiroz \\
  Centro de Inform\'atica\\
 Universidade Federal de Pernambuco\\
  Recife, Brasil \\
  \texttt{ruy@cin.ufpe.br} \\
\And
  Anjolina G. de Oliveira \\
  Centro de Inform\'atica\\
 Universidade Federal de Pernambuco\\
  Recife, Brasil \\
  \texttt{ago@cin.ufpe.br} \\
}

\begin{document}
\maketitle

\begin{abstract}

Using a natural deduction system in the Curry--Howard tradition, we demonstrate how to formalise the concept of computational paths (sequences of rewrites) as equalities between two terms of the same type. The aim is to formulate a term rewriting system in order to illustrate how one can perform computations within these computational paths, establishing equalities between equalities. We shall proceed to use the Labelled Natural Deduction -- LND -- based on the concept of computational paths (which is a system of rewrites) as a tool for obtaining the results on the fundamental group of the circle, the torus and the real projective plane.

\end{abstract}


\section{Introduction}

The identity type is arguably one of the most interesting entities of  Martin-L\"{o}f type theory. From any type $A$, it is possible to construct the identity type $Id_A (x,y)$. This type establishes the relation of identity between two terms of $A$, i.e., if there is  $x =_p y: A$, then $p$ is a witness or proof that $x$ is indeed equal to $y$.  The proposal of the Univalence Axiom made the identity type one of the most studied aspects of type theory. It proposes that $x=y$ is equivalent to saying that $x\simeq y$, that is, the identity is equivalent to equivalence. Another important aspect is the fact that it is possible to interpret the identity t`ype as paths between two points of the same space. This interpretation gives rise to the interesting rendition of equality as a collection of homotopical paths. This connection of type theory and homotopy theory makes type theory a suitable foundation for both computation and mathematics. Nevertheless, this interpretation is only a semantical one  and it has not been proposed with a syntactical counterpart for the concept of path in type theory (cf.\ Voevodsky's observation.\footnote{As it turns out, in the formulation of the identity type, there is only one introduction rule, and thus only one identifier/constructor, i.e.\ `$\mathsf{refl}$', to be used to construct elements of $Id_A(a,b)$. Even though the `$\mathsf{transport}$' constructor allows for the construction of other paths in the syntax, the fact that the only base case is `$\mathsf{refl}$' has brought some mistery as to whether terms for paths are actually integrated in the whole conceptual and formal framework of type theory. From an email message by M.\ Escard\`o:
\begin{quote}
``In the vein of bringing to public record  things that Vladimir said, here is a short interview. \\
-------- Forwarded Message --------\\
Subject: Re: historical question\\
Date: Thu, 22 Oct 2015 16:08:14 -0400\\
From: Vladimir Voevodsky <vlad...@ias.edu>\\
To: Martin Escardo <m.es...@cs.bham.ac.uk>\\
CC: Prof. Vladimir Voevodsky <vlad...@ias.edu>\\

(...)

``> (Also: what was your first reaction when you saw the identity type for
> the first time? Did you immediately connect it with path spaces?)

Not at all. I did not make this connection until late 2009. All the time before it I was hypnotized by the mantra that the only inhabitant of the Id type is reflexivity which made then [sic] useless from my point of view."

Vladimir.
\end{quote}
published in a google discussion group of Homotopy Type Theory on Thursday, 12 October 2017 20:24:26,  with Vladimir Voevodsky in 2015. https://groups.google.com/g/homotopytypetheory/c/K\_4bAZEDRvE/m/VbYjok7bAAAJ}) For that reason, the addition of paths to the syntax of homotopy type theory has been recently proposed by de Queiroz, Ramos and de Oliveira ~\cite{Ruy1,Art1}. The idea is to use an entity named \emph{computational path}, proposed by de Queiroz and Gabbay in 1994  ~\cite{Ruy4}, and demonstrate that it can be useful for formalizing the identity type in a more explicit manner.

On the other hand, one of the main interesting points of the interpretation of logical connectives via deductive systems which use a labelling system is the clear separation between a functional calculus on the labels (the names that record the steps of the proof) and a logical calculus on the formulas ~\cite{Lof1,Ruy4}. Moreover, this interpretation has important applications. The works of ~\cite{Ruy1,Ruy4,Ruy5,RuyAnjolinaLivro} claim that the harmony that comes with this separation makes labelled natural deduction a suitable framework to study and develop a theory of equality for natural deduction. Take, for example, the following cases taken from the $\lambda$-calculus ~\cite{Ruy5}:

\begin{center}
$(\lambda x.(\lambda y.yx)(\lambda w.zw))v \rhd_{\eta} (\lambda x.(\lambda y.yx)z)v \rhd_{\beta} (\lambda y.yv)z \rhd_{\beta} zv$

$(\lambda x.(\lambda y.yx)(\lambda w.zw))v \rhd_{\beta} (\lambda x(\lambda w.zw)x)v \rhd_{\eta} (\lambda x.zx)v \rhd_{\beta} zv$
\end{center}

In the theory of the $\beta\eta$-equality of $\lambda$-calculus, we can indeed say that $(\lambda x.(\lambda y.yx)(\lambda w.zw))v$ is equal to $zv$. Moreover, as can be seen above, there are at least two ways of obtaining these equalities. We can go further, and call $s$ the first sequence of \textit{rewrites} that establish that $(\lambda x.(\lambda y.yx)(\lambda w.zw))v$ is indeed equal to $zv$. The second one, for example, can be called $r$. Thus, we can say that this equality is established by $s$ and $r$. As we shall see in this paper, $s$ and $r$ are examples of an entity known as \textit{computational path}. \footnote{By introducing terms as {\em explicit\/} conversion paths between terms, allowing for further iterations coming from conversions between those paths, we make an attempt at dealing with an important issue brought about by Voevodsky, where by defining a rule stating that `any two terms of an equality type are equivalent under conversions' one would immediately resurrect the undecidability of the identity type:
\begin{quote}
``The introduction of the equality types however brings with it a whole new dimension to the type system because the type expressions $eq_{\mathbf{R}}(r, r')$ unlike the type expressions we have considered before depend on terms. Another issue which arises is how to ensure that the equality types are mapped by models to $pt$ or $\emptyset$ and not to sets with many elements. One can impose the later condition by a rule saying that any two terms of an equality type are equivalent under conversions but this immediately resurrects the un-decidability argument. I do not know of any satisfactory solution of these problems in the usual type theories." \cite{Vlad2} (p.23)
\end{quote}
}

Since we now have labels (computational paths) that establish the equality between two terms, interesting questions might arise: (i) is $s$ different from $r$; (ii) are there normal forms of this equality proof; (iii) if $s$ is equal to $r$, how can we prove this? Questions like these can be answered when we work in a labelled natural deduction framework. The idea is that we are not limited by the calculus on the formulas, but we can also define and work with rules that apply to the labels. That way, we can use these rules to formally establish the equality between these labels, i.e., establish equalities between equalities. Here we shall use a system proposed by ~\cite{Anjo1} and known as $\mathit{LND_{EQ}-TRS}$.\footnote{The formalization of such an iteration between equalities seems to find a counterpart in Voevodsky's observations on the `isomorphism invariance principle':
\begin{quote}
    ``One of the keystones of contemporary mathematics is the \textbf{isomorphism invariance principle}: for any statement $P$ about $X$ and any isomorphism $X\stackrel{\phi}{\simeq}X'$, there is a statement $P_\phi$ about $X'$ such that $P$ holds iff $P_\phi$ holds

The equality problem in formalizations comes in part from the fact that when one encodes $X$ and $X'$ the isomorphism is lost.

There is more to the equality problem than isomorphism invariance:

$\bullet$ equality is a good notion for `elements' -- individuals, but fails for collections.

$\bullet$ isomorphism is a good notion for collections, but fails for collections of collections.

This leads to a theory of iterated $n$-equivalences which are the correct replacements for such ``iterated collection" " (\textit{Foundations of Mathematics and Homotopy Theory}, IAS, 2006)
\end{quote}
A recent presentation by Michael Shulman (\emph{Homotopy type theory - A high-level language for invariant mathematics}, March 2019) draws attention to the importance of Voevodsky's homotopy type theory in finding a proper framework to deal with higher-order equalities in mathematics:
\begin{quote}
``Homotopy type theory is a high-level abstract framework for working with sameness."
\end{quote}
}

In that context, the contribution of this paper will be to propose an interesting connection: it is possible to use a labelled natural deduction system together with $\mathit{LND_{EQ}-TRS}$ to obtain topological results about fundamental groups. 

Indeed, in this paper we will develop a framework for dealing with such expressions representing computational paths and show that it is powerful enough to provide the formal tools to calculate the fundamental group of the circle, the torus and the real projective plane. For this, we use a labelled deduction system based on the concept of computational paths (sequence of rewrites).  Taking into account that in mathematics ~\cite{Munkres} the calculation of this fundamental group is quite laborious, we believe our work accomplishes this calculation in a less complex form. Nevertheless, to obtain this result we first need to formally define the concept of computational paths and define $\mathit{LND_{EQ}-TRS}$.

As a matter of fact, this is part of an ongoing project \cite{ramos0,ramos1,veras4,veras1,veras2,veras3,Martinez-1,Martinez-2,Martinez-3} which, while it looks for the use of homotopy structures such as groupoids in the study of semantics of computation, it also seeks to demonstrate the utility and the impact of the so-called Curry-Howard interpretation of logical deduction in the actual practice of an important area of mathematics, namely homotopy theory. The short citation for the Royal Swedish Academy of Sciences' ``2020 Rolf Schock Prize in logic and philosophy” says that it was awarded to Per Martin-L\"of (shared with Dag Prawitz) ``for the creation of constructive type theory." In a longer statement, the prize committee recalls that constructive type theory is ``a formal language in which it is possible to express constructive mathematics" (…) ``[which] also functions as a powerful programming language and has had an enormous impact in logic, computer science and, recently, mathematics." 

In fact, by introducing a framework whose formalisation of the logical notion of equality is done via the so-called ``identity type", we have the possibility for a surprising connection between term rewriting and geometric concepts such as path and homotopy. And indeed, Martin-L\"of’s type theory (MLTT) allows for making useful bridges between theory of computation, algebraic topology, logic, categories, and higher algebra, and a single concept seems to serve as a bridging bond: ``path". Its impact in mathematics has been felt more strongly since the start of Vladimir Voevodsky’s program on the univalent foundations of mathematics around 2005, and one specific aspect which we would like to address here is the calculation of fundamental groups of surfaces. Taking from the Wikipedia entry on ``homotopy group", calculation of homotopy groups is in general much more difficult than some of the other homotopy invariants learned in algebraic topology. Now, by using an alternative formulation of the ``identity type" which provides an explicit formal account of ``path", operationally understood as an invertible sequence of rewrites (such as Church’s ``conversion"),\footnote{From Martin-L\"of's first published article on type theory with identity types:
\begin{quote}
``The formal system that we shall erect consists of a certain number of mechanical rules for deriving symbolic expressions of the forms
$$a\in A$$
and
$$a\mbox{ conv }b$$
which are to be read $a$ is a \emph{term with type symbol} $A$ and $a$ \emph{converts} into $b$, respectively. In the intended interpretation, $a\in A$ will mean that $a$ is an object of type $A$ and $a$ conv $b$ that $a=_{\rm def} b$." (\cite{Aruy13}, p.87)
\end{quote}
Church's (1936) \cite{church36} original formulation of the `conversion' equality:
\begin{quote}
``We consider the three following operations on a well-formed formula:
\begin{enumerate}
    \item[I.] To replace any part $\lambda x[M]$ of a formula by $\lambda y [S^x_y M|]$ where $y$ is a variable which does not occur in $M$.
    \item[II.] To replace any part of $\{\lambda x[M]\}(N)$ of a formula by $S^x_N|M$, provided that the bound variables in $M$ are distinct both from $x$ and from the free variables in $N$.
    \item[III.] To replace any part of $S^x_N|M$ (not immediately following $\lambda$) of a formula by $\{\lambda x[M]\}(N)$, provided that the bound variables in $M$ are distinct both from $x$ and from the free variables in $N$.
\end{enumerate}
Any finite sequence of these operations is called a \emph{conversion}, and if $B$ is obtainable from $A$ by a conversion we say that $A$ is \emph{convertible} into $B$, or ``$A$ conv $B$".
(...)

A function $F$ of one positive integer is said to be $\lambda$-\emph{definable} if it is possible to find a formula $\mathbf{F}$ such that, if $F(m)=r$ and $\mathbf{m}$ and $\mathbf{r}$ are the formulas for which the integers $m$ and $r$ (written in Arabic notation) stand according to our abbreviations introduced above, then $\{\mathbf{F}\}(\mathbf{m})$ conv $\mathbf{r}$."
\end{quote}
}
and interpreted as a homotopy, we wish to provide examples of calculations of fundamental groups of surfaces such as the circle, the torus, the 2-holed torus, the Klein bottle, and the real projective plane. We would like to suggest that these examples might bear witness to the impact of MLTT in mathematics by offering formal tools to calculate and prove fundamental groups, as well as allowing to make such calculations and proofs amenable to be dealt with by systems of formal mathematics and interactive theorem provers such as Coq, Lean, and similar ones.

\section{Computational paths} \label{cpaths}

In this section, we introduce the main working tool, an entity known as computational paths. In  ~\cite{Art1}, we have seen that it is possible to interpret the identity type semantically, considering the terms as homotopical paths between two points of a space. Thus, inspired by the path-based approach of the homotopy interpretation, we can use a similar approach to define the identity type in type theory, this entity is known as computational path.\footnote{Our approach differs from the one developed in the HoTT book \cite{hottbook} in that we do not need to simulate the path-space between those points via a certain coding process, since we have  \underline{computational paths} as part of the theory's syntax, and they are an integral part of the whole formalisation framework in the tradition of the Curry-Howard interpretation. See, for example:
\begin{quote}
   ``To characterize a path space, the first step is to define a comparison fibration ``$\mathsf{code}$" that provides
a more explicit description of the paths. There are several different methods for proving that such
a comparison fibration is equivalent to the paths (we show a few different proofs of the same
result in \S 8.1). The one we have used here is called the \textbf{encode-decode method}: the key idea is
to define $\mathsf{decode}$ generally for all instances of the fibration (i.e.\ as a function $\Pi_{(x:A+B)}\mathsf{code}(x)\to
(\mathsf{inl}(a_0) = x))$, so that path induction can be used to analyze $\mathsf{decode}(x, \mathsf{encode}(x, p))$." (p.95)
\end{quote}
There have been attempts at developing other methods to improve on the `encode-decode' method, such as \cite{Kraus}:
\begin{quote}
    ``Often, we want to find out what
specific equality types look like. This is directly the goal when calculating
the homotopy groups of given types (as in the synthetic homotopy theory
mentioned above), but it is also a necessary intermediate step for many other
constructions.  For a very concrete example, let us recall the calculation of
the loop space of the circle $\mathbb{S}^1$
 by Licata and Shulman. This loop space
of $\mathbb{S}^1$, as defined above in (4), is by definition simply the equality type
(base = base). Licata and Shulman introduce and explain the encode-decode
method: in order go get started, they ``guess" that the loop space in question
is equivalent to the integers $\mathbb{Z}$ (looking at the left side of (4), the intuition
is that one can go around the loop clockwise any number of times, and
negative numbers correspond to going counterclockwise).

(...)

The encode-decode method has been employed successfully in a variety
of cases. Going through the necessary steps can be somewhat tedious but it
often at least partially mechanical. One main goal in this paper is to develop
a different method to directly work with equality types of coequalizers and
pushouts (and constructions based on them).
\end{quote}
The proposal makes use of a path constructor called `$\mathsf{glue}$', which is built to simulate paths but is not part of the primitive constructions in HoTT, even if, as it is argued, it can be constructed out of the standard elimination operator $\mathtt{J}$ of the identity type.
}

The interpretation will be similar to the homotopy case: a term $p : Id_{A}(a,b)$ will be a computational path between terms $a, b : A$, and such path will be the result of a sequence of rewrites. In what follows, we shall formally define the concept of a computational path. The main notion, i.e.\ proofs of equality statements as (reversible) sequences of rewrites, is not new, as it goes back to a paper entitled "Equality in labeled deductive systems and the functional interpretation of propositional equality", presented in December 1993 at the {\em 9th Amsterdam Colloquium\/}, and published in the proceedings in 1994~\cite{Ruy4}.

Indeed, one of the most interesting aspects of the identity type is the fact that it can be used to construct higher structures. This is a rather natural consequence of the fact that it is possible to construct higher identities. For any $a, b : A$, we have type $Id_{A}(a,b)$. If this type is inhabited by any $p, q:Id_{A}(a,b)$, then we have type $Id_{Id_{A}(a,b)}(p,q)$. If the latter type is inhabited, we have a higher equality between $p$ and $q$~\cite{harper1}. This concept is also present in computational paths. One can prove the equality between two computational paths $s$ and $t$ by constructing a third one between $s$ and $t$. We provide in this section a system of rules used to establish equalities between computational paths~\cite{Anjo1}. 

Another important question we seek to answer is one that arises naturally when talking about equality: Is there a canonical proof for a statement like $t_{1} = t_{2}$? In the language of computational paths, is there a normal path between $t_{1}$ and $t_{2}$ such that every other path can be reduced to it? In ~\cite{Arttese}, it was proved that the answer is negative, this model also refutes the \gls{uip}).

\subsection{Introducing computational paths} \label{path}

Before we go into details on computational paths, let us recall what motivated the introduction of computational paths to type theory. In type theory, our types are interpreted using the so-called Brower-Heyting-Kolmogorov Interpretation. That way, a semantic interpretation of formulas are not given by truth-values, but by the concept of proof as a primitive notion. Thus, we have ~\cite{Ruy1}:

\begin{tabbing}
	\textbf{a proof of the proposition:} \hbox{\ \ \ \ \ } \= \textbf{is given by:} \\
	$A\land B$ \> a proof of $A$ \textbf{and} a proof of $B$ \\ 
	$A\lor B$ \> a proof of $A$ \textbf{or} a proof of $B$ \\
	$A\rightarrow B$ \> a \textbf{function} that turns a proof of $A$ into a proof of $B$ \\
	$\forall x^D.P(x)$ \> a \textbf{function} that turns an element $a$ into a proof of $P(a)$ \\
	$\exists x^D.P(x)$ \> an element $a$ (witness) \textbf{and a proof of} $P(a)$ \\
\end{tabbing}

Also, based on the Curry-Howard functional interpretation of logical connectives, one has ~\cite{Ruy1}:

\begin{tabbing}
	\textbf{a proof of the proposition:} \hbox{\ \ \ \ \ } \= \textbf{has the canonical form of:} \\
	$A\land B$ \> $\langle p,q\rangle$ where $p$ is a proof of $A$ and $q$ is a proof of $B$ \\
	$A\lor B$ \> $i(p)$ where $p$ is a proof of $A$ or $j(q)$ where $q$ is a proof of $B$ \\
	\> (`$i$' and `$j$' abbreviate `into the left/right disjunct') \\
	$A\rightarrow B$ \> $\lambda x.b(x)$ where $b(p)$ is a proof of B \\
	\> provided $p$ is a proof of A \\
	$\forall x^A.B(x)$ \> $\Lambda x.f(x)$ where $f(a)$ is a proof of $B(a)$ \\ 
	\> provided $a$ is an arbitrary individual chosen\\
	\> from the domain $A$\\ 
	$\exists x^A.B(x)$ \> $\varepsilon x.(f(x),a)$ where $a$ is a witness\\
	\> from the domain $A$, $f(a)$ is a proof of $B(a)$ \\
\end{tabbing}
(For a more detailed explanation of the last clause, see \cite{Aruy30}.)

Upon further inspection, there is one interpretation missing in the BHK-Interpretation. What constitutes a proof of $t_{1} = t_{2}$? In other words, what is a proof of an equality statement? In ~\cite{Ruy1} it was proposed that an equality between these two terms should be a sequence of rewritings starting at $t_{1}$ and ending at $t_{2}$.

We answer this by proposing that an equality between those two terms should be a sequence of rewrites starting from $t_{1}$ and ending at $t_{2}$. Thus, we would have ~\cite{Ruy1}:

\begin{tabbing}
	\textbf{a proof of the proposition:} \hbox{\ \ \ \ \ } \= \textbf{is given by:} \\
	\\
	$t_1= t_2$ \> ? \\
	\> (Perhaps a sequence of rewrites \\
	\> starting from $t_1$ and ending in $t_2$?) \\
\end{tabbing}

We call computational path the sequence of rewrites between these terms.

\subsection{Formal definition}

Since computational path is a generic term, it is important to emphasize the fact that we are using the term computational path in the sense defined by~\cite{Ruy5}. A computational path is based on the idea that it is possible to formally define when two computational objects $a,b : A$ are equal. These two objects are equal if one can reach $b$ from $a$ by applying a sequence of axioms or rules. This sequence of operations forms a path. Since it is between two computational objects, it is said that this path is a computational one. Also, an application of an axiom or a rule transforms (or rewrites) a term in another. For that reason, a computational path is also known as a sequence of rewrites. Nevertheless, before we formally define a computational path, we can take a look at one famous equality theory, the $\lambda\beta\eta-equality$~\cite{lambda}:

\begin{definition}[\cite{lambda}]
	The \emph{$\lambda\beta\eta$-equality} is composed by the following axioms:
	
	\begin{enumerate}
		\item[$(\alpha)$] $\lambda x.M = \lambda y.M[y/x]$ \quad if $y \notin FV(M)$;
		\item[$(\beta)$] $(\lambda x.M)N = M[N/x]$;
		\item[$(\rho)$] $M = M$;
		\item[$(\eta)$] $(\lambda x.Mx) = M$ \quad $(x \notin FV(M))$.
	\end{enumerate}
	
	And the following rules of inference:

	\bigskip
	\noindent
	\begin{bprooftree}
		\AxiomC{$M = M'$ }
		\LeftLabel{$(\mu)$ \quad}
		\UnaryInfC{$NM = NM'$}
	\end{bprooftree}
	\begin{bprooftree}
		\AxiomC{$M = N$}
		\AxiomC{$N = P$}
		\LeftLabel{$(\tau)$}
		\BinaryInfC{$M = P$}
	\end{bprooftree}
	
	\bigskip
	\noindent
	\begin{bprooftree}
		\AxiomC{$M = M'$ }
		\LeftLabel{$(\nu)$ \quad}
		\UnaryInfC{$MN = M'N$}
	\end{bprooftree}
	\begin{bprooftree}
		\AxiomC{$M = N$}
		\LeftLabel{$(\sigma)$}
		\UnaryInfC{$N = M$}
	\end{bprooftree}
	
	\bigskip
	\noindent
	\begin{bprooftree}
		\AxiomC{$M = M'$ }
		\LeftLabel{$(\xi)$ \quad}
		\UnaryInfC{$\lambda x.M= \lambda x.M'$}
	\end{bprooftree}\footnote{We adhere strictly to the approach of giving the definitional equalities as the basis for the construction of proof-objects of the identity type, something which may not be followed everywhere. For example, in \cite{Vlad2} (p.25) Voevodsky introduces $\mu$-equality as part of the `rules for the equivalence types':
		\begin{quote}
	$$[\mathbf{smart0}]\displaystyle{
	     {\Gamma\vdash\mathbf{Q} : Type\qquad \Gamma,y : \mathbf{R} \vdash \mathbf{q} : \mathbf{Q}\qquad \Gamma\vdash \mathbf{h} : eq_{\mathbf{R}}(\mathbf{r},\mathbf{r}')}\over 
{ \Gamma\vdash \theta y : \mathbf{R}.(\mathbf{h}, \mathbf{q}) : eq_{\mathbf{Q}}(\mathbf{q}(\mathbf{r}/y), \mathbf{q}(\mathbf{r}'/y))}}
\qquad (15)$$
 
	\end{quote}
		and the same for $\xi$-equality:
	\begin{quote}
	$$[\mathbf{smart2}]\displaystyle{
	     {\Gamma,y : \mathbf{R} \vdash e : eq_{\mathbf{Q}}(\mathbf{q},\mathbf{q}')}\over 
{ \Gamma\vdash ex(e) : eq_{\Pi y:\mathbf{R}.\mathbf{Q}}(\lambda y : \mathbf{R}.\mathbf{q}, \lambda y : \mathbf{R}.\mathbf{q}')}}
\qquad \qquad \qquad(16)$$
 
	\end{quote}

	}
	
	
	
\end{definition}

\begin{definition}[\cite{lambda}]
	$P$ is $\beta$-equal or $\beta$-convertible to $Q$  (notation $P=_\beta Q$)
	iff $Q$ is obtained from $P$ by a finite (perhaps empty)  series of $\beta$-contractions
	and reversed $\beta$-contractions  and changes of bound variables.  That is,
	$P=_\beta Q$ iff \textbf{there exist} $P_0, \ldots, P_n$ ($n\geq 0$)  such that
	$P_0\equiv P$,  $P_n\equiv Q$,
	$(\forall i\leq n-1) (P_i\triangleright_{1\beta}P_{i+1}  \mbox{ or }P_{i+1}\triangleright_{1\beta}P_i  \mbox{ or } P_i\equiv_\alpha P_{i+1}).$
\end{definition}
\noindent (Notice that equality has an \textbf{existential} force, which will be shown in the proof rules for the identity type.\footnote{An anonymous referee has asked ``What does it mean that equality has an existential force?" and the answer is that our rules are aimed at formalising the kind of reasoning embedded in the definition of equality between $\lambda$-terms from Church's original definition of \emph{conversion}, which says that two terms $M$ and $N$ are equal if there is a sequence of applications of the rewriting rules ($\beta$, $\alpha$, $\beta^{-1}$) starting from $M$ and arriving at $N$. As for the elimination rules in natural deduction style for existential-content propositions, which uses a local assumption, it is worth noticing that in the elimination rules for Martin-L\"of's original $Id$ type, the framing reflects the pattern of existential-like elimination:
$$
\displaystyle{
\displaystyle{\ \atop {\exists xP(x)}}\quad
\displaystyle{{[P(t)]}\atop {C}}
\over {C}}
$$
even though the `entity' which is at the center of the existential content does not appear explicitly:
$$\displaystyle{{\displaystyle{\ \atop {a:A\quad b:A\quad c:{Id}_A(a,b)}} \quad
\displaystyle{{[x:A]} \atop {d(x):C(x,x,{\tt r}(x))}} } \over
\displaystyle{{\tt J}(c,d):C(a,b,c)}}{Id}\mbox{-{\it elimination\/}}$$
(provided $[x:A,y:A,z:{Id}_A(x,y)]$ leads to $C(x,y,z)\ type$). Notice the use of a new local assumption `$[x:A]$' being introduced to arrive at a certain unspecified conclusion `$C(x,x,{\tt r}(x))$', with the usual provisions. The lack of an explicit entity witnessing the propositional equality renders the explanation of the `elimination' operator `${\tt J}$' a nontrivial challenge. A rather technical explanation is given by \cite{GambinoGarner} via weak factorization systems.

In an early draft entitled `Notes on homotopy $\lambda$-calculus' \cite{Vlad2}, Voevodsky also uses a rule for the $eq_{\mathbf{R}}$ equality type which `may be considered as an analog of the equality elimination rules in other dependent type systems':
\begin{quote}
    ``Since equivalences can be ``pushed through" all term expressions with the help of rule (15) this implies that we only need an analog of the equality elimination rule for the expressions $\mathbf{Q} = eq_{\mathbf{R}}(x, y)$. This is achieved by our rule (17) which therefore may be considered as an analog of the equality elimination rules in other dependent type systems. From this point of view rule (18) corresponds to the $\beta$-conversion for the equality. We could have introduced a conversion instead of the equivalence $\epsilon(-)$ but this approach allows more flexibility in the models.''
\end{quote}
where the rules (17) and (18) are framed as:
$$\displaystyle{
{\Gamma\vdash  \mathbf{R} : Type} \over
{\Gamma,x,y,z : R,\phi : eq_{\mathbf{R}}(x,y),\psi : eq_{\mathbf{R}}(x,z) \vdash s(\phi,\psi) : eq_{\Sigma u:\mathbf{R}.eq_\mathbf{R}(x,u)}(\langle y,\phi\rangle,\langle z,\psi\rangle)}
}\qquad (17)
$$
$$\displaystyle{
{\Gamma\vdash  \mathbf{R} : Type} \over
{\Gamma, x, y : R, \phi : eq_{\mathbf{R}}(x, y), \vdash \epsilon(\phi) : eq_{eq_{\mathbf{R}}(x,y)}(\pi(s(id(x), \phi)), \phi)}
}\qquad (18)
$$
})

The same happens with $\lambda\beta\eta$-equality:
\begin{definition}[$\lambda\beta\eta$-equality~\cite{lambda}]
	The equality-relation determined by the theory $\lambda\beta\eta$ is called
	$=_{\beta\eta}$; that is, we define
	$$M=_{\beta\eta}N\quad\Leftrightarrow\quad\lambda\beta\eta\vdash M=N.$$
\end{definition}

\begin{example}\normalfont
	Take the term $M\equiv(\lambda x.(\lambda y.yx)(\lambda w.zw))v$. It is $\beta\eta$-equal to $N\equiv zv$ because of the sequence:\\
	$(\lambda x.(\lambda y.yx)(\lambda w.zw))v, \quad  (\lambda x.(\lambda y.yx)z)v, \quad   (\lambda y.yv)z , \quad zv$\\
	which starts from $M$ and ends with $N$, and each member of the sequence is obtained via 1-step $\beta$- or $\eta$-contraction of a previous term in the sequence. To turn this sequence into a {\em path\/}, one has to apply transitivity twice, as we do in the example below.
\end{example}

\begin{example}\label{examplepath} \normalfont
	The term $M\equiv(\lambda x.(\lambda y.yx)(\lambda w.zw))v$ is $\beta\eta$-equal to $N\equiv zv$ because of the sequence:\\
	$(\lambda x.(\lambda y.yx)(\lambda w.zw))v, \quad  (\lambda x.(\lambda y.yx)z)v, \quad   (\lambda y.yv)z , \quad zv$\\
	Now, turning this sequence into a path leads us to the following:\\
	The first is equal to the second based on the grounds:\\
	$\eta((\lambda x.(\lambda y.yx)(\lambda w.zw))v,(\lambda x.(\lambda y.yx)z)v)$\\
	The second is equal to the third based on the grounds:\\
	$\beta((\lambda x.(\lambda y.yx)z)v,(\lambda y.yv)z)$\\
	The first is equal to the third based on the grounds:\\
	$\tau(\eta((\lambda x.(\lambda y.yx)(\lambda w.zw))v,(\lambda x.(\lambda y.yx)z)v),\beta((\lambda x.(\lambda y.yx)z)v,(\lambda y.yv)z))$\\
	The third is equal to the fourth one based on the grounds:\\
	$\beta((\lambda y.yv)z,zv)$\\
	The first one is equal to the fourth one based on the grounds:\\
	$\tau(\tau(\eta((\lambda x.(\lambda y.yx)(\lambda w.zw))v,(\lambda x.(\lambda y.yx)z)v),\beta((\lambda x.(\lambda y.yx)z)v,(\lambda y.yv)z)),\beta((\lambda y.yv)z,zv)))$.
\end{example}


The aforementioned theory establishes the equality between two $\lambda$-terms. Since we are working with computational objects as terms of a type, we can consider the following definition:

\begin{definition}
	The equality theory of Martin L\"of's type theory has the following basic proof rules for the $\Pi$-type:
	
	\bigskip
	
	\noindent
	\begin{bprooftree}
		\hskip -0.3pt
		\alwaysNoLine
		\AxiomC{$N : A$}
		\AxiomC{$[x : A]$}
		\UnaryInfC{$M : B$}
		\alwaysSingleLine
		\LeftLabel{$(\beta$) \quad}
		\BinaryInfC{$(\lambda x.M)N = M[N/x] : B[N/x]$}
	\end{bprooftree}
	\begin{bprooftree}
		\hskip 11pt
		\alwaysNoLine
		\AxiomC{$[x : A]$}
		\UnaryInfC{$M = M' : B$}
		\alwaysSingleLine
		\LeftLabel{$(\xi)$ \quad}
		\UnaryInfC{$\lambda x.M = \lambda x.M' : (\Pi x : A)B$}
	\end{bprooftree}
	
	\bigskip
	
	\noindent
	\begin{bprooftree}
		\hskip -0.5pt
		\AxiomC{$M : A$}
		\LeftLabel{$(\rho)$ \quad}
		\UnaryInfC{$M = M : A$}
	\end{bprooftree}
	\begin{bprooftree}
		\hskip 100pt
		\AxiomC{$M = M' : A$}
		\AxiomC{$N : (\Pi x : A)B$}
		\LeftLabel{$(\mu)$ \quad}
		\BinaryInfC{$NM = NM' : B[M/x]$}
	\end{bprooftree}
	
	\bigskip
	
	\noindent
	\begin{bprooftree}
		\hskip -0.5pt
		\AxiomC{$M = N : A$}
		\LeftLabel{$(\sigma) \quad$}
		\UnaryInfC{$N = M : A$}
	\end{bprooftree}
	\begin{bprooftree}
		\hskip 105pt
		\AxiomC{$N : A$}
		\AxiomC{$M = M' : (\Pi x : A)B$}
		\LeftLabel{$(\nu)$ \quad}
		\BinaryInfC{$MN = M'N : B[N/x]$}
	\end{bprooftree}
	
	\bigskip
	
	\noindent
	\begin{bprooftree}
		\hskip -0.5pt
		\AxiomC{$M = N : A$}
		\AxiomC{$N = P : A$}
		\LeftLabel{$(\tau)$ \quad}
		\BinaryInfC{$M = P : A$}
	\end{bprooftree}
	
	\bigskip
	
	\noindent
	\begin{bprooftree}
		\hskip -0.5pt
		\AxiomC{$M: (\Pi x : A)B$}
		\LeftLabel{$(\eta)$ \quad}
		\RightLabel {$(x \notin FV(M))$}
		\UnaryInfC{$(\lambda x.Mx) = M: (\Pi x : A)B$}
	\end{bprooftree}
	
	\bigskip
	
\end{definition}

We are finally able to formally define computational paths:

\begin{definition}
	Let $a$ and $b$ be elements of a type $A$. Then, a \emph{computational path} $s$ from $a$ to $b$ is a composition of rewrites (each rewrite is an application of the inference rules of the equality theory of type theory or is a change of bound variables). We denote it by $a =_{s} b$.
\end{definition}

As we have seen in \emph{Example \ref{examplepath}}, the composition of rewrites is an application of the rule $\tau$. Since the change of bound variables is possible, each term is considered up to an $\alpha$-equivalence.

\subsection{Equality equations}

One can use the aforementioned axioms to illustrate that computational paths establish the three fundamental equations of equality: the reflexivity, symmetry and transitivity:

\bigskip

\begin{bprooftree}
\AxiomC{$a =_{t} b : A$}
\AxiomC{$b =_{u} c : A$}
\RightLabel{\textit{transitivity}}
\BinaryInfC{$a =_{\tau(t,u)} c : A$}
\end{bprooftree}
\begin{bprooftree}
\AxiomC{$a : A$}
\RightLabel{\textit{reflexivity}}
\UnaryInfC{$a =_{\rho} a : A$}
\end{bprooftree}

\bigskip
\begin{bprooftree}
\AxiomC{$a =_{t} b : A$}
\RightLabel{\textit{symmetry}}
\UnaryInfC{$b =_{\sigma(t)} a : A$}
\end{bprooftree}

\bigskip

\subsection{Identity type}

We have said that one can formulate the identity type using computational paths. As we have seen, the best way to define any formal entity of type theory is by a set of natural deduction rules. Thus, we define our path-based approach as the following set of rules: 

\begin{itemize}
	
\item Formation and Introduction rules ~\cite{Ruy1,Art1}:
	\bigskip
	\begin{center}
		\begin{bprooftree}
			\AxiomC{$A$ type}
			\AxiomC{$a : A$}
			\AxiomC{$b : A$}
			\RightLabel{$Id - F$}
			\TrinaryInfC{$Id_{A}(a,b)$ type}
		\end{bprooftree}
	
	\bigskip
	
		\begin{bprooftree}
			\AxiomC{$a =_{s} b : A$}
			\RightLabel{$Id - I$}
			\UnaryInfC{$s(a,b) : Id_{A}(a,b)$}
		\end{bprooftree}
	\end{center}
	\bigskip
	
	One can notice that our formation rule is exactly equal to the identity type in type theory. From terms $a, b : A$, it is possible to claim that the identity type is inhabited only if there is a proof of equality between those terms, i.e., $Id_{A}(a,b)$.\footnote{An anonymous referee questioned `the distinction between ``$a =_s b : A$" and ``$s:Id_A(a,b)$", and the answer starts from Martin-L\"of's original distinction between propositions and judgements: 
	\begin{quote}
	    ``If $x$ and $y$ are objects of one and the same type $A$, then 
	    
	    \centerline{$I(x, y)$}
	    
is a proposition, namely, the proposition that $x$ and $y$ are \emph{identical}." (p.81)
	\end{quote}
	Further along the same line of reasoning:
	\begin{quote}
	    ``if $a =_{\rm def} b$, then $I(a,b)$ is true, that is, $a$ and $b$ are identical" (p.86)
	\end{quote}
	In our proposed formulation, the propositional equality is written as $Id_A(x,y)$, and judgemental equality as $a=_s b:A$ (where instead of simply saying that the latter comes from a \emph{definitional} equality, it carries an identifier `$s$' denoting the rewriting path, i.e., the composition of (possibly several) definitional equalities.
	
	From `Truth of a proposition, evidence of a judgement, validity of a proof' \cite{martin-lof-prop-judg}:
	\begin{quote}
	    ``First of all, we have the notion of proposition. Second, we have the notion of truth of a proposition. Third, combining these two, we arrive at the notion of assertion or judgement." (p.409)
	\end{quote}
	}
	
	The difference starts with the introduction rule. In our approach, one can notice that we do not use a reflexive constructor $r$. In other words, the reflexive path is not the main building block of our identity type. Instead, if we have a computational path $a =_{s} b : A$, we introduce $s(a,b)$ as a term of the identity type. That way, one should see $s(a,b)$ as a sequence of rewrites and substitutions (i.e., a computational path) which would have started from $a$ and arrived at $b$.
	
	\bigskip 
	
	\item Elimination rule ~\cite{Ruy1,Art1}:
	
	\begin{center}
		\begin{bprooftree}
			\alwaysNoLine
			\AxiomC{$m : Id_{A}(a,b)$ }
			\AxiomC{$[a =_{g} b : A]$}
			\UnaryInfC{$h(g) : C$}
			\alwaysSingleLine
			\RightLabel{$Id - E$}
			\BinaryInfC{$REWR(m, \acute{g}.h(g)) : C$}
		\end{bprooftree}
	\end{center}
	\bigskip
	
	Let us recall the notation being used. First, one should see $h(g)$ as a functional expression $h$ which depends on $g$. Also, one should notice the use of `$\acute{\ }$' in $\acute{g}$. One should see `$\acute{\ }$' as an abstractor that binds the occurrences of the variable $g$ introduced in the local assumption $[a =_{g} b : A]$ as a kind of {\em Skolem-type\/} constant denoting the {\em reason\/} why $a$ was assumed to be equal to $b$.
	
	We also introduce the constructor $REWR$. In a sense, it is similar to the constructor $J$ of the traditional approach, since both arise from the elimination rule of the identity type. The behavior of $REWR$ is simple. If from a computational path $g$ that establishes the equality between $a$ and $b$ one can construct $h(g) : C$, then if we also have this equality established by a term $C$, we can gather all this information in $REWR$ to construct $C$, eliminating the type $Id_{A}(a,b)$ in the process. The idea is that we can substitute $g$ for $m$ in $\acute{g}.h(g)$, resulting in $h(m/g) : C$. This behavior is established below by the reduction rule.\footnote{The use of the variable-binding via a notation which differs from $\lambda$ has the purpose of avoiding imposing each and every property of $\lambda$-conversion for this particular abstraction. Perhaps in the same vein, Voevodsky introduces a $\theta$-abstraction in \cite{Vlad2}:
	\begin{quote}
	    ``The notation $\theta y : \mathbf{R}.(h, q)$ is chosen to emphasize that $y$ becomes a bound variable in this expression." (p.25)
	\end{quote}
	Yet another notation for variable binding (abstraction) is used by Martin-L\"of in the definition of the rule of $\Sigma$-elimination:
	$$\displaystyle{
	{
	\displaystyle{\ \atop{c \in (\Sigma x \in A)B(x)}} \qquad
	\displaystyle{{(x \in A, y \in B(x))} \atop {d(x,y) \in C((x,y))}}
	} \over
	\displaystyle{\mathsf{E}(c, (x, y) d(x, y)) \in C(c)}
	}
	$$
	``Another notation for $\mathsf{E}(c, (x, y) d(x, y))$ could be $(\mathsf{E} x, y) (c, d(x, y))$, but we prefer the first since it shows more clearly that $x$ and $y$ become bound only in $d(x, y)$." \cite{Lof1}(p.40)
	}
	
	\item Reduction rule 	~\cite{Ruy1,Art1}:
	
	\bigskip
	\begin{center}
		\begin{bprooftree}
			\AxiomC{$a =_{m} b : A$}
			\RightLabel{$Id - I$}
			\UnaryInfC{$m(a,b) : Id_{A}(a,b)$}
			\alwaysNoLine
			\AxiomC{$[a =_{g} b : A]$}
			\UnaryInfC{$h(g) : C$}
			\alwaysSingleLine
			\RightLabel{$Id - E$ \quad $\rhd_\beta$}
			\BinaryInfC{$REWR(m, \acute{g}.h(g)) : C$}
		\end{bprooftree}
		\begin{bprooftree}
			\AxiomC{$a =_{m} b : A$}
			\alwaysNoLine
			\UnaryInfC{$h(m/g):C$}
		\end{bprooftree}
	\end{center}
	\bigskip
	
	\item Induction rule: 
	
	\bigskip
	\begin{center}
		\begin{bprooftree}
			\AxiomC{$e : Id_{A}(a,b)$}
			\AxiomC{$[a =_{t} b : A]$}
			\RightLabel{$Id - I$}
			\UnaryInfC{$t(a, b) : Id_{A}(a, b)$}
			\RightLabel{$Id - E$ \quad  $\rhd_{\eta}$ \quad $e : Id_{A}(a,b)$}
			\BinaryInfC{$REWR(e, \acute{t}.t(a,b)) : Id_{A}(a,b)$}
		\end{bprooftree}
	\end{center}
	\bigskip
	
\end{itemize}

Our introduction and elimination rules reassure the concept of equality as an \textbf{existential force}. In the introduction rule, we encapsulate the idea that a witness of an identity type $Id_{A}(a,b)$ only exists if there is a computational path establishing the equality of $a$ and $b$. Also, one can notice that this elimination rule is similar to the elimination rule of the existential quantifier.

\subsection{Path-based examples}

The objective of this subsection is to demonstrate how to put into practice the rules that we have just defined. The objective is to show the construction of terms of some important types. The constructions that we have chosen to build are the reflexive, transitive and symmetric type of the identity type. Those were not random choices. The main reason for having picked them is the fact that reflexive, transitive and symmetric types are essential to the process of building a groupoid model for the identity type~\cite{hofmann1}. As we shall see, these constructions come naturally from simple computational paths constructed by the application of axioms of the equality of type theory.

Before we start the constructions, it is our opinion that it is essential to understand how to use the eliminations rules. The process of building a term of some type is a matter of finding the right reason. In the case of $J$, the reason is the correct $x,y : A$ and $z : Id_{A}(a,b)$ that generates the adequate $C(x,y,z)$. In our approach, the reason is the correct path $a =_{g} b$ that generates the adequate $g(a,b) : Id(a,b)$.

\subsubsection{Reflexivity}

One could find strange the fact that we need to prove the reflexivity. Nevertheless, just remember that our approach is not based on the idea that reflexivity is the base of the identity type. As usual in type theory, a proof of something comes down to a construction of a term of a type. In this case, we need to construct a term of type $\Pi_{(a : A)}Id_{A}(a,a)$. The reason is extremely simple: from a term $a : A$, we obtain the computational path $a =_{\rho} a : A$ ~\cite{Art1}:

\bigskip

\begin{center}
	\begin{bprooftree}
		\AxiomC{$[a : A]$}
		\UnaryInfC{$a =_{\rho} a : A$}
		\RightLabel{$Id - I$}
		\UnaryInfC{$\rho(a,a) : Id_{A}(a,a)$}
		\RightLabel{$\Pi-I$}
		\UnaryInfC{$\lambda a.\rho(a,a) : \Pi_{(a : A)}Id_{A}(a,a)$}
	\end{bprooftree}
\end{center}

\bigskip

\subsubsection{Symmetry}

The second proposed construction is the symmetry. Our objective is to obtain a term of type $\Pi_{(a : A)}\Pi_{(b : A)}(Id_{A}(a,b)\rightarrow Id_{A}(b,a))$.

We shall construct a proof using computational paths. As expected, we need to find a suitable reason. Starting from $a =_{t} b$, we could look at the axioms of \emph{Definition 4.1} to plan our next step. One of those axioms makes the symmetry clear: the $\sigma$ axiom. If we apply $\sigma$, we will obtain $b =_{\sigma(t)} a$. From this, we can then infer that $Id_A$ is inhabited by $(\sigma(t))(b,a)$. Now, it is just a matter of applying the elimination ~\cite{Art1}:

\bigskip

\begin{center}
	\begin{bprooftree}
		\alwaysNoLine
		\AxiomC{$[a:A] \quad [b:A]$}
		\UnaryInfC{$[p(a,b) : Id_{A}(a,b)]$}
		\alwaysSingleLine
		\AxiomC{[$a =_{t} b : A$]}
		\UnaryInfC{$b =_{\sigma(t)} a : A$}
		\RightLabel{$Id - I$}
		\UnaryInfC{$(\sigma(t))(b,a) : Id_{A}(b,a)$}
		\RightLabel{$Id - E$}
		\BinaryInfC{$REWR(p(a,b),\acute{t}.(\sigma(t))(b,a)) : Id_{A}(b,a)$}
		\RightLabel{$\rightarrow - I$}
		\UnaryInfC{$\lambda p.REWR(p(a,b), \acute{t}.(\sigma(t))(b,a)) : Id_{A} (a,b) \rightarrow Id_{A}(b,a)$}
		\RightLabel{$\Pi-I$}
		\UnaryInfC{$\lambda b. \lambda p.REWR(p(a,b),\acute{t}.(\sigma(t))(b,a)) :  \Pi_{(b : A)}(Id_{A} (a,b) \rightarrow Id_{A}(b,a))$}
		\RightLabel{$\Pi-I$}
		\UnaryInfC{$\lambda a.\lambda b. \lambda p.REWR(p(a,b), \acute{t}.(\sigma(t))(b,a)) :  \Pi_{(a : A)}\Pi_{(b : A)}(Id_{A} (a,b) \rightarrow Id_{A}(b,a))$}
	\end{bprooftree}
\end{center}

\bigskip

\subsubsection{Transitivity}
The third and last construction will be the transitivity. Our objective is to obtain a term of type   $$\Pi_{(a : A)}\Pi_{(b : A)}\Pi_{(c : A)} (Id_{A}(a,b) \rightarrow Id_{A}(b,c) \rightarrow Id_{A}(a,c)).$$

To build our path-based construction, the first step, as expected, is to find the reason. Since we are trying to construct the transitivity, it is natural to think that we should start with paths $a =_{t} b$ and $b =_{u} c$ and then, from these paths, we should conclude that there is a path $z$ that establishes that $a =_{z} c$. To obtain $z$, we could try to apply the axioms of \emph{Definition 4.1}. Looking at the axioms, one of them is exactly what we want: the axiom $\tau$. If we apply $\tau$ to  $a =_{t} b$ and $b =_{u} c$, we will obtain a new path $\tau(t,u)$ such that $a = _{\tau(t,u)} c$. Using that construction as the reason, we obtain the following term ~\cite{Art1}:

\begin{center}
	\begin{figure}
		\begin{sideways}
		\begin{bprooftree}
			\alwaysNoLine
			\AxiomC{$[a:A] \quad [b:A]$}
			\UnaryInfC{$[w(a,b) : Id_{A}(a,b)]$}
			\alwaysNoLine
			\AxiomC{$[c:A]$}
			\UnaryInfC{$[s(b,c) : Id_{A}(b,c)]$}
			\alwaysSingleLine
			\AxiomC{$[a =_{t} b:A]$}
			\AxiomC{$[b =_{u} c:A]$}
			\BinaryInfC{$a =_{\tau(t,u)} c:A$}
			\RightLabel{$Id - I$}
			\UnaryInfC{$(\tau (t,u))(a,c) : Id_{A}(a,c)$}
			\RightLabel{$Id - E$}
			\BinaryInfC{$REWR(s(b,c),\acute{u}(\tau (t,u))(a,c)) : Id_{A}(a,c)$}
			\RightLabel{$Id - E$}
			\BinaryInfC{$REWR(w(a,b),\acute{t}REWR(s(b,c),\acute{u}(\tau (t,u))(a,c))) : Id_{A}(a,c)$}
			\RightLabel{$\rightarrow - I$}
			\UnaryInfC{$\lambda s.REWR(w(a,b),\acute{t}REWR(s(b,c),\acute{u}(\tau (t,u))(a,c))) : Id_{A}(b,c) \rightarrow Id_{A}(a,c)$}
			\RightLabel{$\rightarrow - I$}
			\UnaryInfC{$\lambda w.\lambda s.REWR(w(a,b),\acute{t}REWR(s(b,c),\acute{u}(\tau (t,u))(a,c))) : Id_{A}(a,b) \rightarrow Id_{A}(b,c) \rightarrow Id_{A}(a,c)$}
			\RightLabel{$\Pi-I$}
			\UnaryInfC{$\lambda c.\lambda w.\lambda s.REWR(w(a,b),\acute{t}REWR(s(b,c),\acute{u}(\tau (t,u))(a,c))) :  \Pi_{(c : A)}(Id_{A}(a,b) \rightarrow Id_{A}(b,c) \rightarrow Id_{A}(a,c))$}
			\RightLabel{$\Pi-I$}
			\UnaryInfC{$\lambda b. \lambda c.\lambda w.\lambda s.REWR(w(a,b),\acute{t}REWR(s(b,c),\acute{u}(\tau (t,u))(a,c))) :  \Pi_{(b : A)}\Pi_{(c : A)}(Id_{A}(a,b) \rightarrow Id_{A}(b,c) \rightarrow Id_{A}(a,c))$}
			\RightLabel{$\Pi-I$}
			\UnaryInfC{$\lambda a. \lambda b. \lambda c.\lambda w.\lambda s.REWR(w(a,b),\acute{t}REWR(s(b,c),\acute{u}(\tau (t,u))(a,c))) :   \Pi_{(a : A)}\Pi_{(b : A)}\Pi_{(c : A)}(Id_{A}(a,b) \rightarrow Id_{A}(b,c) \rightarrow Id_{A}(a,c))$}
		\end{bprooftree}
	\end{sideways}
	\end{figure}
	\end{center}

\newpage

As one can see, each step is simply composed of straightforward applications of introduction, elimination rules and abstractions. The only idea behind this construction is the simple fact that the axiom $\tau$ guarantees the transitivity of paths.

\subsection{Term rewrite system}

As we have just shown, a computational path establishes when two terms of the same type are equal. From the theory of computational paths, an interesting case arises. Suppose we have a path $s$ that establishes that $a =_{s} b : A$ and a path $t$ that establishes that $a =_{t} b : A$. Consider that $s$ and $t$ are formed by distinct compositions of rewrites. Is it possible to conclude that there are cases where $s$ and $t$ should be considered equivalent? The answer is \emph{yes}. Consider the following examples ~\cite{Arttese}:

\begin{example}
	\noindent \normalfont Consider the path  $a =_{t} b : A$. By the symmetric property, we obtain $b =_{\sigma(t)} a : A$. What if we apply the property again on the $\sigma(t)$ path? We would obtain a path  $a =_{\sigma(\sigma(t))} b : A$. Since we applied the symmetry twice in succession, we have obtained a path that is equivalent to the initial path $t$. For that reason, we would like to conclude that the act of applying the symmetry twice in succession is a redundancy. We say that the path $\sigma(\sigma(t))$ reduces to path $t$.
\end{example}

\begin{example}

	\noindent \normalfont Consider the reflexive path $a =_{\rho} a : A$. If the symmetric axiom is applied, we end up with $a =_{\sigma(\rho)} a : A$. Thus, the obtained path is equivalent to the initial one, since the symmetry was applied to the reflexive path. Therefore, $\sigma(\rho)$ is a redundant way of expressing the path $\rho$. Thus, $\sigma(\rho)$ should be reduced to $\rho$.
	
\end{example}

\begin{example}
	\noindent \normalfont Consider a path $a =_{t} b : A$. By applying the symmetry, one ends up with $b =_{\sigma(t)} a : A$. It is possible to take those two paths and apply the transitivity, ending up with $a =_{\tau(t,\sigma(t))} a$. Since the path $\tau$ is the inverse of the $\sigma(\tau)$, the composition of those two paths should be equivalent to the reflexive path. Thus, $\tau(t,\sigma(t))$ should be reduced to $\rho$.
\end{example}	
	
As can be seen in the aforementioned examples, different paths should be considered equal if one is a redundant form of the other. The examples that we have just seen are straightforward and simple cases. Since the equality theory has a total of 7 axioms, the possibility of combinations which could generate redundancies is high. Fortunately, all possible redundancies were thoroughly mapped by~\cite{Anjo1}. In that work, a system that establishes all redundancies and creates rules which solve them was proposed. This system, known as $\mathit{LND_{EQ}-TRS}$, maps a total of $39$ rules that solve redundancies. 

\subsection{$\mathit{\textbf{LND}_{\textbf{EQ}}\textbf{-TRS}}$}

In this subsection, we provide the rules which compose the $\mathit{LND_{EQ}-TRS}$. All those rules originate from the mapping of redundancies between computational paths, as we have seen in the $3$ previous examples. 

\subsubsection{Subterm substitution}

Before we introduce the rewriting rules, it is important to introduce the concept of subterm substitution. In Equational Logic, the subterm substitution is given by the following inference rule~\cite{Ruy2}:
\begin{center}
	\begin{bprooftree}
		\AxiomC{$s = t$ }
		\UnaryInfC{$s\theta = t\theta$}
	\end{bprooftree}
\end{center}
 where $\theta$ is a substitution.
One problem is that such rule does not respect the sub-formula property. To deal with that,~\cite{Chenadec} proposes two inference rules:

\begin{center}
	\begin{bprooftree}
		\AxiomC{$M = N$}
		\AxiomC{$C[N] = O$}
		\RightLabel{$IL$ \quad}
		\BinaryInfC{$C[M] = O$}
	\end{bprooftree}
	\begin{bprooftree}
		\AxiomC{$M = C[N]$}
		\AxiomC{$N = O$}
		\RightLabel{$IR$ \quad}
		\BinaryInfC{$M = C[O]$}
	\end{bprooftree}
\end{center}

\noindent where M, N and O are terms.

As proposed in ~\cite{Ruy1}, we can define similar rules using computational paths, as follows:

\begin{center}
	\begin{bprooftree}
		\AxiomC{$x =_r {\cal C}[y]: A$}
		\AxiomC{$y =_s u : A'$}
		\BinaryInfC{$x =_{{\tt sub}_{\tt L}(r,s)} {\cal C}[u]: A$}
	\end{bprooftree}
	\begin{bprooftree}
		\AxiomC{$x =_r w : A'$}
		\AxiomC{${\cal C}[w]=_s u : A$}
		\BinaryInfC{${\cal C}[x]=_{{\tt sub}_{\tt R}(r,s)} u : A$}
	\end{bprooftree}
\end{center}

\noindent where $C$ is the context in which the sub-term detached by '[ ]' appears and $A'$ could be a sub-domain of $A$, equal to $A$ or disjoint to $A$.

In the rule above, ${\cal C}[u]$ should be understood as the result of replacing every occurrence of $y$ by $u$ in $C$.\footnote{In Martin-L\"of's original paper \cite{Aruy13} (p.85) the rule for subterm substitution is framed as:
$${{a=_{\rm def}c}\over {b[a]=_{\rm def}b[c]} }$$
}

\subsubsection{Rewriting rules}

In this subsection, our objective is to demonstrate all rewrite reductions and their associated rewriting rules. The aim is to analyze all possible occurrences of redundancies in proofs which involve the rules of rewritings. 

We start with the transitivity:

\begin{definition}[reductions involving $\tau$ ~\cite{Ruy1}]
	
	\bigskip
	
	$$\displaystyle{ x =_r y : {A} \quad y =_{\sigma(r)} x :
		{A} \over x =_{\tau(r,\sigma(r))} x : {A}} \quad 
	\triangleright_{tr}  \quad x =_{\rho} x : {A}$$
	\bigskip
	$$\displaystyle{ y =_{\sigma (r)} x : {A} \quad x=_r y 
		:  {A} \over y =_{\tau(\sigma(r),r)} y : {A}} \quad \triangleright_{tsr} \quad  y=_{\rho} y : {A}$$
	\bigskip
	$$\displaystyle{ u=_r v : {A} \quad v=_{\rho} v :  
		{A} \over u =_{\tau(r,\rho)} v : {A}} \quad 
	\triangleright_{trr}  \quad  u =_r v : {A}$$
	\bigskip
	$$\displaystyle{ u=_{\rho} u : {A} \quad u=_r v : 
		{A} \over u=_{\tau(\rho, r)} v : {A}} \quad 
	\triangleright_{tlr}  \quad u=_r v : {A}$$
	\bigskip
	
	Associated rewriting rules:
	
	$$\tau(r,\sigma(r))\triangleright_{tr}\rho$$
	$$\tau(\sigma(r),r) \triangleright_{tsr}\rho$$
	$$\tau(r,\rho) \triangleright_{trr} r$$
	$$\tau(\rho, r) \triangleright_{tlr} r.$$

\end{definition} 

These reductions can be generalized to transformations where the reasons $r$
and $\sigma(r)$ (transf.\ 1 and 2) and $r$ and $\rho$ (transf.\ 3 and 4) appear
in some context, as illustrated by the following example:
 ~\cite{Ruy1}:

\begin{example}\ 
	
	\bigskip
	
	$\displaystyle{\displaystyle{\displaystyle{ \atop x=_r y :{A}} \over
			{i} (x) =_{\xi_1 (r)} i(y) : {A} + {B}} \quad 
		\displaystyle{\displaystyle{ x =_r y : {A} \over y=_{\sigma (r)} x :
				{A}} \over {i} (y) =_{\xi_1(\sigma (r))} {i} (x) : {A}
			+ {B}} \over {i} (x) =_{\tau ( \xi_1 (r), \xi_1 (\sigma 
			(r)))} {i} (x) : {A} + {B}} $
	
	\hfill{$\triangleright_{tr} \quad \displaystyle{{x=_r y:A}\over{i(x)=_{\xi_1(r)}i(y):A+B}}$}
	\ \\
	\bigskip
	
	\noindent Associated rewriting:
	 $\tau ( \xi_1 (r), \xi_1 (\sigma (r))) \triangleright_{tr} \xi_1 (r).$
	
\end{example}

\medskip

\noindent For the general context ${\cal C}[\ ]$:

\noindent Associated rewritings:\\
$\tau({\cal C}[r] , {\cal C}[\sigma(r)] ) \triangleright_{tr} {\cal C}[\rho]$\\
$\tau({\cal C}[\sigma(r)] , {\cal C}[r]) \triangleright_{tsr} {\cal C}[\rho]$\\
$\tau({\cal C}[r], {\cal C}[\rho]) \triangleright_{trr} {\cal C}[r]$\\
$\tau({\cal C}[\rho], {\cal C}[r]) \triangleright_{tlr} {\cal C}[r]$

\medskip

The transitivity rules are quite straightforward. There are some more complicated cases in ~\cite{Ruy1}:

\begin{definition}\ 
	
	\bigskip
	
	$\displaystyle{\displaystyle{\displaystyle{\displaystyle{ \atop } \atop 
			}  \atop a: { A} } \quad \displaystyle{\displaystyle{ 
				\displaystyle{ {[ x: {A} ]} \atop \vdots} \atop { b(x) =_r g(x) : {B}}}
			\over \lambda x.b(x) =_{\xi ( r)} \lambda x . g(x) : { A} \to {B}}\to \mbox{\it -intr} \over
		{APP} ( \lambda x. b(x) , a) =_{ \nu ( \xi ( r))} {APP}  
		(\lambda x. g(x) , a) : {B}}\to\mbox{\it -elim}$
	
	\smallskip
	
	\hfill{$\triangleright_{mxl} \quad \displaystyle{{a:A}\atop{b(a/x) =_{r} g(a/x) : {B}}}$}\\
	\
	\bigskip
	
	Associated rewriting rule:
	$$ \nu ( \xi ( r)) \triangleright_{mxl} r.$$
\end{definition}

\medskip

\begin{definition}[reductions involving $\rho$ and $\sigma$ ~\cite{Ruy1}]
	
	\bigskip
	
	$$ \displaystyle{ x=_\rho x : {A} \over x =_{\sigma(\rho)} x : {A}}  \quad \triangleright_{sr} \quad  x =_\rho x : {A}$$
	\bigskip
	$$ \displaystyle{\displaystyle{ x=_r y : {A} \over y =_{\sigma(r)
			} x : {A}} \over x =_{\sigma(\sigma(r)) }y:A}  \quad \triangleright_{sr} \quad
	x=_r y : {A}$$
	
	\bigskip
	Associated rewritings:\\
	$\sigma(\rho)\triangleright_{sr} \rho$ \\
	$\sigma(\sigma(r))\triangleright_{sr} r$ \\
\end{definition}

\medskip

\begin{definition}[Substitution rules ~\cite{Ruy1}]
	
	\bigskip
	
	$$\displaystyle{ u =_r {\cal C}[x]:{A} \quad x =_{\rho} x : {A'} \over u =_{{\tt sub_L}(r,\rho)}  {\cal C}[x] : {A}} \quad \triangleright_{slr} \quad u =_r {\cal C}[x]: {A}$$
	\bigskip
	$$\displaystyle{ x=_{\rho} x : {A'} \quad {\cal C}[x] =_r z : {A} \over {\cal C}[x] =_{{\tt sub_R}(\rho,r)} z :
		{A}} \quad \triangleright_{srr} \quad {\cal C}[x] =_r z : {A}$$
	\bigskip
	$$\displaystyle{\displaystyle{ z =_s {\cal C}[y] : {A}
			\quad y =_r w : {A'} \over z =_{{\tt sub_L} (s,r)} {\cal C}[w] :
			{D}} \quad \displaystyle{ y =_r w :{A'} \over w =_{\sigma(r)} y :
			{D'}} \over z =_{{\tt sub_L} ({\tt sub_L} ( s,r) , \sigma (r))} {\cal C}[y] : {A}} \  \triangleright_{sls} \  z =_s {\cal C}[y] : {A}$$
	\bigskip
	$$\displaystyle{\displaystyle{  z =_{s} {\cal C}[y] : { A}  \quad   y =_{r} w : {A'} \over z =_{{\tt sub_L} (  s , r)} {\cal C}[w] : {A}} \quad
		\bigskip
		 \displaystyle{ y =_r w : {A'} \over w=_{\sigma (r)} y : {A'}} \over z =_{{\tt sub_L} ( {\tt sub_L} ( s, r) ,
			\sigma(r))} {\cal C}[y] : {A}} \; \triangleright_{slss} \;  z =_{s} {\cal C}[y]  : {A}$$
	\bigskip
	$$\displaystyle{ \displaystyle{\displaystyle{ \atop } \atop
			x=_s w:  {A'}} \quad  \displaystyle{ \displaystyle{ x=_s w : {A'}
				\over  w =_{\sigma (s)} x : {A'}} \quad \displaystyle{ \atop {\cal C}[x] =_{r} z : {A}} \over {\cal C}[w] =_{{\tt sub_R} ( \sigma  (s) , r)} z : {A}} \over  {\cal C}[x] =_{{\tt sub_R}  (s,
			{\tt sub_R} ( \sigma (s), r))} z : {A}} \; \triangleright_{srs} \;  {\cal C}[x]  =_{r} z : {A}$$
	\bigskip
	$$\displaystyle{\displaystyle{ x=_s w : {A'} \over w
			=_{\sigma (s)} x : {A'}} \quad \displaystyle{ x=_s w : {A'}
			\quad {\cal C}[w] =_r z : {A} \over {\cal C}[x] =_{{\tt sub_R}
				(s,r)}  z : {A}} \over {\cal C}[w] =_{{\tt sub_R} ( \sigma (s)
			,  {\tt sub_R} (s,r))} z : {A}} \ \triangleright_{srrr} \ {\cal C}[w] =_r
	z :  {A}$$
	
	\bigskip \noindent
	Associated rewritings:\\
	${\tt sub_L}({\cal C}[r] , {\cal C}[\rho]) \triangleright_{slr} {\cal C}[r]$\\
	${\tt sub_R}({\cal C}[\rho], {\cal C}[r]) \triangleright_{srr} {\cal C}[r]$\\
	${\tt sub_L} ( {\tt sub_L} ( s ,{\cal C}[r])  , {\cal C}[\sigma(r)]) \triangleright_{sls} s$\\
	${\tt sub_L} ({\tt sub_L} (s, {\cal C}[\sigma(r)]) , {\cal C}[r] ) \triangleright_{slss} s$\\
	${\tt sub_R} (s, {\tt sub_R} ({\cal C}[\sigma(s)], r)) \triangleright_{srs} r$\\
	${\tt sub_R} ({\cal C}[\sigma(s)], {\tt sub_R} ({\cal C}[s], r)) \triangleright_{srrr} r$\\
\end{definition}

\begin{definition}[~\cite{Ruy1}]\ \\
	
	\bigskip
	
	$\beta_{rewr}$-$\times$-{\it reduction}
	
	\noindent $\displaystyle{\displaystyle{x=_r y : { A}  \qquad z : { B}
			\over
			\langle x,z \rangle =_{\xi_1 (r)} \langle y,z \rangle : {A} \times {B}
		}\times \mbox{{\it -intr}}
		\over
		{FST}( \langle x,z \rangle ) =_{\mu_1 ( \xi_1 ( r))} {FST}(\langle y,z 
		\rangle ) : {A} 
	}\times \mbox{{\it -elim}}$
	
	\hfill{$\triangleright_{mx2l} \quad x =_r y : {A}$}
	
	\bigskip
	
	\noindent $\displaystyle{\displaystyle{x =_rx': { A}  \qquad y=_s z : { B}
			\over
			\langle x,y \rangle =_{\xi_\land (r,s)} \langle x',z \rangle : {A} \times {B}
		}\times \mbox{{\it -intr}}
		\over
		{FST}( \langle x,y \rangle ) =_{\mu_1 ( \xi_\land( r,s))} {FST}(\langle x',z 
		\rangle ) : {A} 
	}\times \mbox{{\it -elim}}$
	
	\hfill{$\triangleright_{mx2l} \quad x =_r x' : {A}$}
	
	\bigskip
	
	\noindent $\displaystyle{\displaystyle{x=_r y : {A}  \qquad z=_sw  : {B}
			\over
			\langle x,z \rangle =_{\xi_\land (r,s)} \langle y,w \rangle : {A} \times {B}
		}\times \mbox{{\it -intr}}
		\over
		{SND} (\langle x, z \rangle ) =_{\mu_2 ( \xi_\land( r,s))} {SND} (\langle
		y,w \rangle ) : {B}
	}\times \mbox{\it -elim}$
	
	\hfill{$\triangleright_{mx2r} \quad z =_s w : {B}$}
	
	\bigskip
	
	\noindent $\displaystyle{\displaystyle{x: {A}  \qquad z =_s w : {B}
			\over
			\langle x,z \rangle =_{\xi_2 (s)} \langle x,w \rangle : {A} \times {B}
		}\times \mbox{{\it -intr}}
		\over
		{SND} (\langle x, z \rangle ) =_{\mu_2 ( \xi_2 ( s))} {SND} (\langle
		x,w \rangle ) : {B}
	}\times \mbox{\it -elim}$
	
	\hfill{$\triangleright_{mx2r} \quad z =_s w : {B}$}
	
	\bigskip\noindent
	Associated rewritings: \\
	$\mu_1 ( \xi_1 ( r))\triangleright_{mx2l1} r$\\
	$\mu_1 ( \xi_\land (r,s))\triangleright_{mx2l2} r$\\
	$\mu_2 ( \xi_\land( r,s))\triangleright_{mx2r1} s$\\
	$\mu_2 ( \xi_2 ( s))\triangleright_{mx2r2} s$
	
	\bigskip
	
	\noindent $\beta_{rewr}$-$+$-{\it reduction}
	
	\noindent $\displaystyle{{\displaystyle{{a =_r a':{A}} \over 
				{i}(a) =_{\xi_1 (r)} {i}(a'):{A} + {B}}+\mbox{
				\it -intr\/} \ 
			\displaystyle{{[x:{A}]} \atop {f(x) =_s k(x):{C}}} \ 
			\displaystyle{ {[y: {B}]} \atop { g(y) =_u h(y):{C}}}} \over
		{{D}({i}(a),\acute{x}f(x), \acute{y}g(y)) =_{\mu
				(\xi_1 (r),s,u)} {D}({i}(a'),\acute{x}k(x), \acute{y}h(y)):{C}}}+\mbox{\it -elim}$
	
	\hfill{$ \triangleright_{mx3l} \quad
		\displaystyle{{a =_r a':{A}} \atop {f(a/x) =_s k(a'/x):{C}}}$}
	
	\bigskip

\noindent $\displaystyle{{\displaystyle{{b =_r b':{B}} \over {j}(b) =_{\xi_2 (r)} {j} (b'):{A} + {B}}+\mbox{\it -intr\/} \ 
		\displaystyle{{[x:{A}]} \atop {f(x) =_s k(x):{C}}} \ 
		\displaystyle{ {[y: {B}]} \atop { g(y) =_u h(y):{C}}}} \over
	{{D}({j}(b),\acute{x}f(x), \acute{y}g(y)) =_{\mu
			(\xi_2 (r),s,u)} {D}({j}(b'),\acute{x}k(x), \acute{y}h(y)):{C}}}+\mbox{\it -elim}$

\hfill{$\triangleright_{mx3r} \qquad \displaystyle{b =_s b':{B} \atop g(b/y) =_u h(b'/y):{C}}$}

\bigskip

\noindent
Associated rewritings:\\
$\mu ( \xi_1 ( r) , s,u)\triangleright_{mx3l} s$\\
$\mu ( \xi_2 ( r), s,u)\triangleright_{mx3r} u$

\bigskip

\noindent $\beta_{rewr}$-$\Pi$-{\it reduction}

\noindent $\displaystyle{\displaystyle{\displaystyle{ \atop } \atop a : {A}} 
	\quad  \displaystyle{\displaystyle{ [ x : {A} ] \atop f(x) =_r g(x) :  
			{B} (x)} \over \lambda x.f(x) = _{\xi (r)} \lambda x.g(x) : 
		\Pi  x:{A}.{B} (x)} \over {APP} (\lambda x.f(x), a) 
	=_{\nu (\xi (r))} {APP} (\lambda x.g(x) , a) : {B} (a)}$

\hfill{$\triangleright_{mxl} \quad \displaystyle{  a: {A}
		\atop f(a/x) =_r g(a/x) : {B} (a)}$}
	
\bigskip

\noindent Associated rewriting:\\
$\nu (\xi (r)) \triangleright_{mxl} r$\\

\bigskip

\noindent $\beta_{rewr}$-$\Sigma$-{\it reduction}

\noindent $\displaystyle{\displaystyle{ a=_r a': {A} \quad f(a) : {B}(a) 
		\over  \varepsilon x.(f(x),a) =_{\xi_1 (r)} \varepsilon x.(f(x),a') : 
		\Sigma  x :{A}.{B} (x)} \quad \displaystyle{ [ t : { A}, 
		g(t) :  {B} (t) ] \atop d(g,t) =_s h(g,t) : { C}} \over {E}  
	(\varepsilon x.(f(x),a), \acute{g} \acute{t} d(g,t)) = _{\mu (\xi_1 ( r) 
		, s) } {E}(\varepsilon x.(f(x),a') , \acute{g} \acute{t}h(g,t)) : 
	{ C}}$

\hfill{$\triangleright_{mxr}  
	\quad \displaystyle{  a=_r a': {A} \quad f(a) : {B} (a)  
		\atop  d(f/g , a/t ) =_s h(f/g , a'/t):{ C}}$}

\bigskip

\noindent $\displaystyle{\displaystyle{ a: {A} \quad f(a)=_r i(a) : {B}(a) 
		\over  \varepsilon x.(f(x),a) =_{\xi_2 (r)} \varepsilon x.(i(x),a) : 
		\Sigma  x :{A}.{B} (x)} \quad \displaystyle{ [ t : { A}, 
		g(t) :  {B} (t) ] \atop d(g,t) =_s h(g,t) : { C}} \over {E}  
	(\varepsilon x.(f(x),a), \acute{g} \acute{t} d(g,t)) = _{\mu (\xi_2 ( r) 
		, s) } {E}(\varepsilon x.(i(x),a) , \acute{g} \acute{t}h(g,t)) : 
	{ C}}$

\hfill{$\triangleright_{mxl}  
	\quad \displaystyle{  a: {A} \quad f(a) =_r i(a) : {B} (a)  
		\atop  d(f/g , a/t ) =_s h(i/g , a/t):{ C}}$}

\bigskip

\noindent Associated rewritings:\\
$\mu (\xi_1 (r) , s) \triangleright_{mxr} s$\\
$\mu (\xi_2 (r) , s) \triangleright_{mxl} s$

\end{definition}

\begin{definition}[$\eta_{rewr}$ ~\cite{Ruy1}]\ \\
	
	\bigskip
	
	$\eta_{rewr}$- $\times$-{\it reduction}
	
	$\displaystyle{\displaystyle{ x=_r y : {A} \times {B} 
			\over
			{FST}(x) =_{\mu_1 ( r)} {FST}(y) : {A}} \times\mbox{\it 
			-elim}  \  \displaystyle{ x=_r y : {A} \times {B} \over
			{SND}(x) =_{\mu_2 ( r)} {SND}(y) : {B}
		}\times\mbox{\it -elim} 
		\over
		\langle {FST}(x) , {SND}(x) \rangle =_{\xi ( \mu_1 ( r),\mu_2(r))}  \langle 
		{FST} (y), {SND} (y) \rangle : {A} \times {B}
	}\times\mbox{\it -intr} $
	
	\hfill{$\displaystyle{\triangleright_{mx}} \   \displaystyle{x =_r y:  {A} \times {B}}$}
	
	\bigskip
	
	\noindent  $\eta_{rewr}$- $+$-{\it reduction}

	\noindent $\displaystyle{\displaystyle{ \atop c=_t d : {A} +  
			{B}}  \displaystyle{ [a_1 =_r a_2 : {A} ] \over {i}  (a_1) =_{\xi_1 (r)} {i} (a_2) : {A} +{B}}+\mbox{\it -intr} \displaystyle{ [ b_1 =_s b_2: {B}
			] \over {j}(b_1) =_{\xi_2 (s)} {j}(b_2) : {A} + 
			{B}}+\mbox{\it -intr} \over {D}(c, \acute{a_1}{i}(a_1), \acute{b_1}{j}(b_1)) =_{\mu ( t, \xi_1 (r) , \xi_2 (s))}
		{D}(d, \acute{ a_2}{i}(a_2), \acute{b_2} {j}(b_2))}+\mbox{\it -elim}$
	
	\hfill{$\triangleright_{mxx} \quad c=_t d : {A} +{B}$}
	
	\bigskip
	
	\noindent $\Pi$-$\eta_{rewr}$-{\it reduction}
	
	\noindent $\displaystyle{\displaystyle{[t:{A}] \quad c =_r d: \Pi x:{ A}.{B} (x) \over {APP}(c,t) =_{\nu (r)} {APP} (d,t): {
				B}(t)} \Pi\mbox{\it -elim} \over \lambda t.{APP}(c,t) =_{\xi (\nu 
			(r)) } \lambda t.{APP} (d,t) : 
		\Pi t:{A}.{B}(t)}\Pi\mbox{\it -intr}$
	
	\hfill{$\triangleright_{xmr} \qquad c=_r d: \Pi x:{A}.{B}(x)$}\\
	where $c$ and $d$ do not depend on $x$.
	
	\bigskip
	
	\noindent $\Sigma$-$\eta_{rewr}$-{\it reduction}
	
	\noindent $\displaystyle{\displaystyle{ \atop c=_sb: \Sigma x:{A}.{B} (x)} 
		\quad  \displaystyle{ [t:{A}] \quad [g(t) =_r h(t) : {B}(t)] 
			\over  \varepsilon y.(g(y),t) =_{\xi_2 (r) } \varepsilon y.(h(y), t) :  
			\Sigma y :{A}.{B} (y)}\Sigma\mbox{\it -intr} \over {E} 
		(c,  \acute{g}\acute{t}\varepsilon y.(g(y),t)) =_{\mu (s, \xi_2 (r))} 
		{E} (b, \acute{h} \acute{t} \varepsilon y.(h(y) , t)) : 
		\Sigma y :{A}.{B} (y)}\Sigma\mbox{\it -elim}$
	
	\hfill{$\triangleright_{mxlr} \quad c =_s b : \Sigma x:{A}.{B} (x)$}
	
	\bigskip
	
	\smallskip\noindent
	Associated rewritings:\\
	$\xi (\mu_1 (r),\mu_2(r)) \triangleright_{mx} r$\\
	$\mu (t, \xi_1 (r) , \xi_2 (s) ) \triangleright_{mxx} t$\\
	$\xi (\nu (r)) \triangleright_{xmr} r$\\
	$\mu ( s , \xi_2 (r)) \triangleright_{mxlr} s$
	
\end{definition}

\medskip

\begin{definition}[$\sigma$ and $\tau$  ~\cite{Ruy1}]\ \\
	
	\bigskip
	
	$\displaystyle{\displaystyle{x=_r y : {A} \quad y =_s w : {A}  
			\over x =_{\tau(r,s)}  w : {A}} \over w=_{\sigma(\tau(r,s ))} x : {A}} \quad \triangleright_{stss} \quad 
	\displaystyle{\displaystyle{  y=_s w : {A} \over w=_{\sigma(s)} y 
			:  {A}} \quad \displaystyle{x=_r y : {A} \over y=_{\sigma(r)} 
			x :  {A}} \over w=_{\tau(\sigma(s),\sigma(r))} x : {A}}$\\
	\ \\
	
	\bigskip
	
	Associated rewriting:\\
	$\sigma(\tau(r,s)) \triangleright_{stss} \tau(\sigma(s), \sigma(r))$
	
\end{definition}

\medskip

\begin{definition}[$\sigma$ and ${\tt sub}$ ~\cite{Ruy1}]\ \\
	
	\bigskip
	
	$$\displaystyle{\displaystyle{ x =_r {\cal C}[y] : {A} \quad y =_s w:  {A'} \over x=_{{\tt sub_L}(r,s)}{\cal C}[w] : {A}} \over 
		{\cal C}[w] =_{\sigma( {\tt sub_L}(r,s))} x : {A}} \quad 
	\triangleright_{ssbl}  \quad \displaystyle{\displaystyle{y=_s w : {A'} \over w=_{\sigma(s)} y : {A'}} \quad   
		\displaystyle{x =_r {\cal C}[y]:{A} \over {\cal C}[y] =_{\sigma(r)} x : {A}} \over {\cal C}[w] =_{{\tt sub_R}(\sigma(s), 
			\sigma(r))} x : {A}}$$
	
	\bigskip
	
	$$\displaystyle{\displaystyle{ x=_r y : {A'} \quad {\cal C}[y] =_s w:  {A} \over {\cal C}[x] = _{{\tt sub_R}(r,s)} w : {A}} \over 
		w=_{ \sigma({\tt sub_R}(r,s))} {\cal C}[x] : {D}} \quad \triangleright_{ssbr} 
	\quad  \displaystyle{\displaystyle{{\cal C}[y] =_s w : {A} \over w =_{\sigma(s)} {\cal C}[y] : 
			{A} } \quad \displaystyle{x=_r y : {A'} \over y=_{\sigma(r)} x 
			:  {A'}} \over w=_{{\tt sub_L}(\sigma(s),\sigma(r))} {\cal C}[x]  : {A}} $$
	
	\bigskip
	
	Associated rewritings:\\
	$\sigma({\tt sub_L}(r,s)) \triangleright_{ssbl} {\tt sub_R}(\sigma(s),  \sigma(r))$ \\
	$\sigma ({\tt sub_R} (r,s)) \triangleright_{ssbr} {\tt sub_L} ( \sigma (s) ,  \sigma (r))$
	
\end{definition}

\begin{definition}[$\sigma$ and $\xi$ ~\cite{Ruy1}]\ \\
	
	\bigskip
	
	$$\displaystyle{\displaystyle{ x=_r y : {A} \over {i}(x) =_{\xi_1 
				(r)}  {i}(y) : {A} +{B}} \over {i}(y) =_{\sigma(\xi_1  (r))} {i}(x) : {A} +{B}} \quad \triangleright_{sx} 
	\quad  \displaystyle{\displaystyle{ x=_r y : 
			{A} \over y =_{\sigma(r)} x : {A}} \over {i}(y) =_{\xi_1 
			(  \sigma (r))} {i} (x) : {A} +{B}}$$
	\bigskip
	$$\displaystyle{\displaystyle{ x=_r y : {A} \quad z=_s w : {B} 
			\over  \langle x,z \rangle =_{\xi(r,s)} \langle y,w \rangle : {A} 
			\times   {B}} \over \langle y,w \rangle =_{\sigma(\xi( r,s))} 
		\langle  x,z \rangle : {A} \times {B}} \quad \triangleright_{sxss} \quad  
	\displaystyle{\displaystyle{x=_r y : {A} \over y=_{\sigma(r)} x :  {A}}  \quad \displaystyle{ z=_s w : {B} \over
			w =_{\sigma(s)} z : {B}} \over \langle y,w \rangle  =_{\xi(\sigma(r),{\sigma(s))}} \langle x,z \rangle : {A} \times {B}}$$
	\bigskip
	
	$$\displaystyle{\displaystyle{\displaystyle{ [ x : {A} ] \atop f(x) =_s 
				g(x)  : {B} (x)} \over \lambda x.f(x) =_{\xi (s)} \lambda x.g(x) :  
			\Pi x:{A}.{B} (x)} \over   \lambda x.g(x) =_{\sigma (\xi (s))} \lambda x.f(x) : \Pi x:{A}.{B} (x)}  \  
	\triangleright_{smss}  \  \displaystyle{\displaystyle{\displaystyle{ [ x : {A} 
				]  \atop f(x) =_s g(x) : {B} (x)} \over  g(x) =_{\sigma(s)} f(x) 
			:  {B} (x)} \over \lambda x.g(x) =_{\xi ( \sigma (s))} \lambda x.f(x) 
		: \Pi x:{A}.{B} (x)} $$
	
	\bigskip
	
	\noindent Associated rewritings:\\
	$\sigma(\xi (r)) \triangleright_{sx} \xi ( \sigma(r))$\\
	$\sigma(\xi (r, s)) \triangleright_{sxss} \xi ( \sigma(r), \sigma(s))$\\
	$\sigma(\xi (s) \triangleright_{smss} \xi ( \sigma(s))$\\
\end{definition}

\begin{definition}[$\sigma$ and  $\mu$ ~\cite{Ruy1}]\ \\
	
	\bigskip
	
	$$\displaystyle{\displaystyle{ x =_r y : {A} \times {B} \over {FST}(x) =_{\mu_1 (r)} {FST} (y) : {A}} \over {FST}(y) =_{\sigma (\mu_1 (r))} {FST} (x) : {A}} \quad
	\triangleright_{sm} \quad   
	\displaystyle{\displaystyle{ x=_r y : {A} \times {B} \over  
			y=_{\sigma(r)} x : {A} \times {B}} \over {FST} (y) =_{\mu_1 
			(\sigma (r))} {FST}(x) : {A}}$$
	\bigskip
	$$\displaystyle{\displaystyle{ x =_r y : {A} \times {B} \over {SND}(x) =_{\mu_2 (r)} {SND} (y) : {A}} \over {SND}(y) =_{\sigma (\mu_2 (r))} {SND} (x) : {A}} \quad
	\triangleright_{sm} \quad   
	\displaystyle{\displaystyle{ x=_r y : {A} \times {B} \over  
			y=_{\sigma(r)} x : {A} \times {B}} \over {SND} (y) =_{\mu_2 
			(\sigma (r))} {SND}(x) : {A}}$$
	
	\bigskip
	
	\noindent $\displaystyle{\displaystyle{ x=_s y : {A} \quad f=_r g : {A} 
			\to  {B} \over {APP} (f,x) =_{\mu (s,r)} {APP }(g,y) : {B}}  \over {APP} (g,y) =_{\sigma(\mu (s,r))} {APP}(f,x) : {B}}$
	
	\hfill{$\triangleright_{smss} \quad \displaystyle{\displaystyle{  x=_s y : {A} 
				\over   y=_{\sigma (s)}  x : {A}} \quad \displaystyle{ f=_r g : 
				{A } \to {B} \over g=_{\sigma(r)} f : {A} \to {B}} 
			\over  {APP}(g,y) =_{\mu (\sigma(s), \sigma(r))} {APP} 
			(f,x)  : {B} }$}
	
	\bigskip
	
	\noindent $\displaystyle{\displaystyle{\displaystyle{\displaystyle{ \atop } 
				\atop    x=_r y : {A} +{B}} \quad 
			\displaystyle{\displaystyle{   [s: {A} ] \atop \vdots} \atop  d(s) =_u 
				f(s)  : {C} } \quad \displaystyle{\displaystyle{ [t:{B}] \atop
					\vdots } \atop  e(t) =_v g(t) : {C}} \over {D} (x, \acute{s}d(s),  \acute{t} e(t)) =_{\mu (r,u,v)} {D}(y, \acute{s}f(s),  \acute{t}g(t)) : {C}} \over {D} (y, \acute{s}f(s),  \acute{t}g(t)) : {C} =_{\sigma(\mu (r,u,v))} {D}  (x, \acute{s}d(s), \acute{t}e(t)) : {C}}$
	
	\hfill{$\triangleright_{smsss} \displaystyle{\displaystyle{\displaystyle{ \atop x=_r y : {A} +  
					{B}} \over y=_{\sigma (r)} x : {A} +{B}} \quad  
			\displaystyle{\displaystyle{[s:{A} ] \atop d(s) =_u f(s) : {C}} 
				\over  f(s) =_{\sigma(u)} d(s) : {C}} \quad  
			\displaystyle{\displaystyle{ [t:{B}] \atop e(t) =_v g(t) : {C} }  
				\over g(t) =_{\sigma (v)} e(t) : {C}} \over {D}  
			(y,\acute{s}f(s), \acute{t}g(t)) =_{\mu (\sigma(r), \sigma
				(u),  \sigma (v))} {D} (x, \acute{s}d(s), \acute{t} e(t)) 
			:  {C}}$}
	
	\bigskip
	
	\noindent $\displaystyle{\displaystyle{\displaystyle{ \atop e=_s b : \Sigma x:{A}.{B} (x)} \quad \displaystyle{ [t: {A}, \; g(t) : {B} (t) 
				]  \atop d(g,t) =_r f(g,t) : {C}} \over {E} (e,\acute{g}\acute{t} d(g,t)) =_{\mu (s,r)} {E} (b, \acute{g} \acute{t} f(g,t)) :  
			{C}} \over {E} (b,\acute{g} \acute{t} f(g,t)) =_{\sigma
			(\mu (s,r))} {E} (e, \acute{g} \acute{t} d(g,t)) : {C}}$
	
	\hfill{$ \triangleright_{smss} \displaystyle{\displaystyle{ \displaystyle{ \atop e=_s b : \Sigma 
					x:{A}.{B} (x)} \over b=_{\sigma(s)} e : \Sigma x:{A}.{B} (x)}\quad \displaystyle{\displaystyle{ [t: {A}, \; g(t) :  
					{B} (t) ] \atop d(g,t) =_r f(g,t) : {C}} \over f(g,t)  =_{\sigma(r)} d(g,t) : {C}} \over {E} (b,\acute{g} \acute{t} f(g,t))  =_{\mu (\sigma (s), \sigma (r))} {E} (e, \acute{g} \acute{t}d(g,t)) : {C}}  $}
	
	\bigskip\noindent
	Associated rewritings:\\
	$\sigma(\mu_1 (r)) \triangleright_{sm} \mu_1 ( \sigma(r))$\\
	$\sigma(\mu_2 (r)) \triangleright_{sm} \mu_2 ( \sigma(r))$\\
	$\sigma(\mu (s, r)) \triangleright_{smss} \mu ( \sigma(s), \sigma(r))$\\
	$\sigma(\mu (r, u,v)) \triangleright_{smsss} \mu ( \sigma(r),  \sigma(u),\sigma(v))$
	
\end{definition}

\medskip

\begin{definition}[$\tau$ and ${\tt sub}$ ~\cite{Ruy1}]\ \\
	
	\bigskip
	
	$\displaystyle{\displaystyle{ x=_r {\cal C}[y]:A \quad y =_s w: {A'} \over x =_{{\tt sub_L} (r,s)} {\cal C}[w]: {A}} \quad
		\displaystyle{ \atop {\cal C}[w] =_t z : { A}} \over x =_{ \tau  ({\tt sub_L} (r,s) , t)} z : {A}}$
	
	\hfill{$  \triangleright_{tsbll} \ \  \displaystyle{\displaystyle{  \atop  
				x=_r {\cal C}[y]: {A}} \quad  \displaystyle{ y =_s w : {A'}  
				\quad {\cal C}[w] =_t z : {A} \over {\cal C}[y] =_{{\tt 
						sub_R}(s,t)}  z : {A}} \over x=_{\tau (r, {\tt sub_R}(s,t))} z 
			:  {A}}$}
	
	\bigskip
	
	\noindent $\displaystyle{\displaystyle{ y=_s w: {A} \quad {\cal C}[w] =_t z 
			:  {A} \over {\cal C}[y] =_{{\tt sub_R} ( s,t)} z : { A}} \quad
		\displaystyle{ \atop z =_u v : { A}} \over {\cal C}[y] =_{ \tau  ({\tt sub_R} (s,t) , u)} v : {A}}$
	
	\hfill{$ \triangleright_{tsbrl} \   \displaystyle{ \displaystyle{ \atop y =_s w : {D'}} \ \  
			\displaystyle{{\cal C}[w] =_t z : {A} \quad z =_u v : {A} \over 
				{\cal C}[w] =_{\tau(t,u)} v : {A}} \over {\cal C}[y]=_{{\tt sub_R } ( s, \tau ( t,u))} v : {A}}$}
	
	\bigskip
	
	\noindent $\displaystyle{\displaystyle{ \atop x=_r {\cal C}[z] : {A}} \quad
		\displaystyle{{\cal C}[z] =_{\rho} {\cal C}[z] : {A} 
			\quad  z =_s w : {A'} \over {\cal C}[z] =_{{\tt sub_L} (\rho , s)} {\cal C}[w] : {A}} \over x =_{\tau ( r, {\tt 
				sub_L}  (\rho, s))} {\cal C}[w] : {A}}$
	
	\hfill{$ \triangleright_{tsblr} \   \displaystyle{ x=_r {\cal C}[z] : {A}  \quad z =_s w : {A'}  \over x =_{{\tt sub_L} (r,s)} {\cal C}[w]: {A}}$}
	
	\bigskip
	
	$\displaystyle{\displaystyle{ \atop x =_r {\cal C}[w]: {A}}  \quad
		\displaystyle{ w=_s z : {A'} \quad {\cal C}[z]=_{\rho} {\cal C}[z] : {A} \over {\cal C}[w] =_{{\tt sub_R} (s, \rho)} {\cal C}[z] : {A}} \over x =_{\tau(r, {\tt sub_R}  
			(s, \rho))} {\cal C}[z] : {A}}$
	
	\hfill{$ \triangleright_{tsbrr} \   \displaystyle{ x =_r {\cal C}[w] : {D} \quad w =_s z : {A'}  
			\over x =_{{\tt sub_L} (r,s)} {\cal C}[z] : {A}}$}
	
\end{definition}

\medskip

\begin{definition}[$\tau$ and $\tau$ ~\cite{Ruy1}]\ \\
	
	\bigskip
	$\displaystyle{\displaystyle{ x=_t y:A \quad y =_r w: {A} \over x =_{\tau(t,r)}w: {A}} \quad
		\displaystyle{ \atop w =_s z : {A}} \over x =_{\tau (\tau (t,r) , s)} z : {A}}$
	
	\hfill{$\triangleright_{tt} \quad  
		\displaystyle{\displaystyle{ \atop x =_t y : {A}} \quad  
			\displaystyle{ y=_r w : {A} \quad w=_s z : {A} \over  y=_{\tau(r,s)} z : {A}} \over x=_{\tau(t,\tau(r,s))} z :  
			{A}}$}
	
	\bigskip\noindent
	Associated rewritings: \\
	$\tau({\tt sub_L}(r,s),t) \triangleright_{tsbll} \tau (r, {\tt sub_R}(s,t))$\\
	$\tau ({\tt sub_R } (s,t), u)) \triangleright_{tsbrl} {\tt sub_R} ( s, \tau  (t,u))$\\
	$\tau (r, {\tt sub_L} (\rho , s)) \triangleright_{tsblr} {\tt sub_L} (r,s)$\\
	$\tau (r, {\tt sub_R} (s, \rho )) \triangleright_{tsbrr} {\tt sub_L} (r,s)$\\
	$\tau(\tau (t,r),s) \triangleright_{tt} \tau(t,\tau(r,s))$
	
\end{definition}

\bigskip

Thus, we put together all those rules to compose our rewrite system:

\begin{definition}[$\mathit{LND_{EQ}-TRS}$ ~\cite{Ruy1}] 
	 \quad \\
1. $\sigma(\rho) \triangleright_{sr} \rho$ \\ 
2. $\sigma(\sigma(r)) \triangleright_{ss} r$\\ 
3. $\tau({\cal C}[r] , {\cal C}[\sigma(r)]) \triangleright_{tr}  {\cal C }[\rho]$\\ 
4. $\tau({\cal C}[\sigma(r)], {\cal C}[r]) \triangleright_{tsr} {\cal C}[\rho]$\\ 
5. $\tau({\cal C}[r], {\cal C}[\rho]) \triangleright_{trr} {\cal C}[r]$\\ 
6. $\tau({\cal C}[\rho], {\cal C}[r]) \triangleright_{tlr} {\cal C}[r]$ \\ 
7. ${\tt sub_L}({\cal C}[r], {\cal C}[\rho]) \triangleright_{slr} {\cal C}[r]$\\ 
8. ${\tt sub_R}({\cal C}[\rho], {\cal C}[r]) \triangleright_{srr} {\cal C}[r]$ \\
9. ${\tt sub_L} ({\tt sub_L} (s, {\cal C}[r]), {\cal C}[\sigma(r)]) \triangleright_{sls} s$\\
10. ${\tt sub_L} ( {\tt sub_L} (s , {\cal C}[\sigma(r)]) , {\cal C}[r]) \triangleright_{slss} s$\\ 
11. ${\tt sub_R} ({\cal C}[s], {\tt sub_R} ({\cal C}[\sigma(s)],r)) \triangleright_{srs} r$\\ 
12. ${\tt sub_R} ({\cal C}[\sigma(s)], {\tt sub_R} ({\cal C}[s] ,  r )) \triangleright_{srrr} r$\\ 
13. 
$\mu_1 ( \xi_1 ( r))\triangleright_{mx2l1} r$\\
14. $\mu_1 ( \xi_\land ( r,s))\triangleright_{mx2l2} r$\\
15.
$\mu_2 ( \xi_\land ( r,s))\triangleright_{mx2r1} s$\\
16.
$\mu_2 ( \xi_2 ( s))\triangleright_{mx2r2} s$\\
17. 
$\mu ( \xi_1 (r) , s , u) \triangleright_{mx3l} s$\\ 
18. 
$\mu (\xi_2 (r) , s , u) \triangleright_{mx3r} u$\\ 
19.
$\nu (\xi (r)) \triangleright_{mxl} r$\\ 
20.
$\mu (\xi_2 (r) , s) \triangleright_{mxr} s$\\ 
21.
$\xi ( \mu_1 (r),\mu_2(r) ) \triangleright_{mx} r$ \\ 
22.
$\mu ( t, \xi_1 (r), \xi_2 (s)) \triangleright_{mxx} t$ \\ 
23. 
$\xi ( \nu (r) ) \triangleright_{xmr} r$ \\ 
24. 
$\mu (s,\xi_2 (r)) \triangleright_{mx1r} s$\\ 
25. $\sigma(\tau(r,s)) \triangleright_{stss} \tau(\sigma(s),  \sigma(r))$\\ 
26. $\sigma({\tt sub_L}(r,s)) \triangleright_{ssbl} {\tt sub_R}(\sigma(s), \sigma(r))$\\ 
27. $\sigma ({\tt sub_R} (r,s)) \triangleright_{ssbr} {\tt sub_L} (\sigma
(s),  \sigma (r))$\\ 
28. $\sigma(\xi (r)) \triangleright_{sx} \xi ( \sigma(r))$\\ 
29. $\sigma(\xi (s, r)) \triangleright_{sxss} \xi ( \sigma(s),  \sigma(r))$\\ 
30. $\sigma(\mu (r)) \triangleright_{sm} \mu ( \sigma(r))$\\ 
31. $\sigma(\mu (s, r)) \triangleright_{smss} \mu (\sigma(s),  \sigma(r))$\\ 
32. $\sigma(\mu (r,u,v)) \triangleright_{smsss} \mu ( \sigma(r),\sigma(u),\sigma(v))$\\
33. $\tau (r, {\tt sub_L} (\rho , s)) \triangleright_{tsbll} {\tt sub_L}  (r,s)$\\ 
34. $\tau (r, {\tt sub_R} (s, \rho)) \triangleright_{tsbrl}  {\tt 
	sub_L} (r,s)$\\ 
35. $\tau({\tt sub_L}(r,s),t) \triangleright_{tsblr} \tau (r, {\tt 
	sub_R} (s,t))$\\ 
36. $\tau ({\tt sub_R} (s,t),u) \triangleright_{tsbrr} {\tt sub_R} (s, \tau  (t,u))$\\ 
37. $\tau(\tau(t,r),s) \triangleright_{tt} \tau(t,\tau (r,s)) $\\
38. $\tau ({\cal C}[u], \tau ({\cal C}[\sigma(u)] , v)) \triangleright_{tts} v$\\
39. $\tau ({\cal C}[\sigma(u)] , \tau ({\cal C}[u] , v)) \triangleright_{tst} u$.
\label{LND-TRS}
\end{definition}

\subsection{Normalization}

	In the previous subsection, we have seen a system of rewrite rules that resolves reductions in a computational path. When we talk about these kinds of systems, two questions emerge: Does every computational path have a normal form? And if a computational path has a normal form, is it unique? To demonstrate that it has a normal form, one must prove that every computational path terminates, i.e., that after a finite number of rewrites, one will end up with a path that does not have any additional reduction. To prove that it is unique, one must illustrate that the system is confluent. In other words, if one has a path with 2 or more reductions, they must demonstrate that the choice of the rewrite rule does not matter. In the end, one will always obtain the same end-path without any redundancies.

\subsubsection{Termination}

We are interested in the following theorem ~\cite{RuyAnjolinaLivro, Ruy1}:

\begin{theorem}[Termination property for $\mathit{LND_{EQ}-TRS}$]
	$\mathit{LND_{EQ}-TRS}$ is terminating.
\end{theorem}

The proofs uses a special kind of ordering, known as \textit{recursive path ordering}, proposed by ~\cite{dershowitz}:

\begin{definition}[Recursive path ordering ~\cite{dershowitz, Ruy1}]
	Let $>$ be a partial ordering 
	on a set of operators F. The recursive path ordering $>^*$ on the set 
	T(F) of terms over F is defined recursively as follows:
	
	$$ s = f(s_1,\ldots , s_m) >^* g(t_1,\ldots, t_n) = t,$$
	if and only if
	\begin{enumerate}
		\item $f=g$ and $\{ s_1, \ldots , s_m\} \gg^*  \{ t_1, \ldots , t_n\}$, or
		\item $f>g$ and $\{s\} \gg^* \{t_1, \ldots , t_n\}$, or
		\item $f \ngeq g$ and $\{s_1, \ldots , s_m \} \gg^* $ or $=$ $\{t\}$
	\end{enumerate}
	where $\gg^*$ is the extension of $>^*$ to multisets.
	
\end{definition}

This definition uses the notion of partial ordering in multisets.  A 
given partial ordering $>$ on a set $S$ may be extended to a partial 
ordering $\gg$ on finite multisets of elements of $S$, wherein a multiset 
is reduced by removing one or more elements and replacing them with any 
finite number of elements, each one smaller than one of the 
elements removed ~\cite{dershowitz}.

Thus, one can prove the termination property by demonstrating that in all rules $e \rightarrow d$ of the system, one has that $e  >^{*} d$. We also need to define the precedence ordering on the rewrite operators. We define it as follows ~\cite{Ruy1, RuyAnjolinaLivro}:
$$\begin{array}{l}
\sigma > \tau > \rho, \\
\sigma > \xi, \\
\sigma > \xi_\land, \\
\sigma > \xi_1, \\
\sigma > \xi_2, \\
\sigma > \mu, \\
\sigma > \mu_1, \\
\sigma > \mu_2, \\
\sigma > {\tt sub_L}, \\
\sigma > {\tt sub_R}, \\
\tau > {\tt sub_L}
\end{array}$$

Thus, one can prove the termination by evidencing that for every rule of $e \rightarrow d$ of $\mathit{LND_{EQ}-TRS}$, $e  >^{*} d$. For almost every rule this is a straightforward and tedious process.  We are not going to display all those steps in this work, but we can give the proof of two examples.

\begin{itemize}
    \item[26.] $\sigma(sub_{L}(r,s))>^{*}sub_{R}(\sigma(s),\sigma(r)):$
    \begin{itemize}
        \item $\sigma>sub_{R}$ from the precedence ordering on the rewrite operators.
        \item $\{\sigma(sub_{L}(r,s))\}\gg^{*}\{\sigma(r),\sigma(r)\}:$
        \begin{itemize}
            \item [-] $\sigma(sub_{L}(r,s))>^{*}\sigma(s)$ and $\sigma(sub_{L}(r,s))>^{*}\sigma(r)$:
            \begin{itemize}
                \item $\sigma=\sigma$
                \item $\{sub(r,s)\}\gg\{s\}$ from the subterm condition.
                \item $\{sub(r,s)\}\gg\{r\}$ from the subterm condition.
            \end{itemize}
           
        \end{itemize}
    \end{itemize}
     
    \item [27.] $\sigma(sub_{R}(r,s))\rhd sub_{L}(\sigma(s),\sigma(r)):$
    \begin{itemize}
        \item $\sigma>sub_{L}$ from the precedence ordering on the rewrite operators.
        \item $\{\sigma(sub_{R}(r,s))\}\gg^{*}\{\sigma(r),\sigma(r)\}:$
            \begin{itemize}
                \item $\sigma=\sigma$
                \item $\{sub_{R}(r,s)\}\gg\{s\}$ from the subterm condition.
                \item $\{sub_{R}(r,s)\}\gg\{r\}$ from the subterm condition.
            \end{itemize}
           
        \end{itemize}
\end{itemize}

All other proofs can be verified at ~\cite{RuyAnjolinaLivro}. 
\subsubsection{Confluence}

Before we go to the proof of confluence, one needs to observe that $\mathit{LND_{EQ}-TRS}$ is a conditional term rewriting system. This means that some rules can only be applied if the terms of the associated equation follow some rules. For example, for the rule  $\mu_1 ( \xi_\land ( r,s))\triangleright_{mx2l2} r$, it is necessary to have a $\beta$-Reduction such as $FST\langle x, y \rangle$. With that in mind, we have the following definition ~\cite{RuyAnjolinaLivro}:

\begin{definition}[Conditional term rewriting system]
	In conditional term rewriting systems, the rules have conditions attached, which must be true for the rewrite to occur. For example, a rewrite rule $e \rightarrow d$ with condition $C$ is expressed as:
	
	\begin{center}
		$C | e \rightarrow d$
	\end{center}
	
	\end{definition}

To prove the confluence, it is necessary to analyze all possible critical pairs using the superposition algorithm proposed by ~\cite{knuth1}. Thus, there should not be any divergent critical pair. For example, we can take the superposition of rules $1$ and $2$, obtaining: $\sigma(\sigma(\rho))$. We have two possible rewrites ~\cite{RuyAnjolinaLivro}:

\begin{itemize}
	\item $\sigma(\sigma(\rho)) \rhd_{sr} \sigma(\rho) \rhd_{sr} \rho$
	\item $\sigma(\sigma(\rho)) \rhd_{ss} \rho$.
\end{itemize}

As can be seen, we ended up with the same term $\rho$. Thus, no divergence has been generated.

One should compare every pair of rules to find all critical pairs and see if there are any divergences. If some divergence occurs, the superposition algorithm proposed by ~\cite{knuth1} illustrates how to add new rules to the system in such a way that it becomes confluent. As a matter of fact, that was the reason why rules $38$ and $39$ of $\mathit{LND_{EQ}-TRS}$ have been introduced to the system ~\cite{Ruy1}:\\

\noindent 38. $\tau ({\cal C}[u], \tau ({\cal C}[\sigma(u)] , v)) \triangleright_{tts} v$\\
39. $\tau ({\cal C}[\sigma(u)] , \tau ({\cal C}[u] , v)) \triangleright_{tst} u$.\\

Those two rules introduced the following reductions to the system ~\cite{RuyAnjolinaLivro}:

\bigskip

\begin{center}
	
\begin{bprooftree}
	\AxiomC{$x =_{s} u : D$}
	\AxiomC{$ x =_{s} u : D$}
	\UnaryInfC{$ u =_{\sigma(s)} x : D$}
	\AxiomC{$x =_{v} w : D$}
	\BinaryInfC{$u =_{\tau(\sigma(s), v)} w : D$}
	\RightLabel{\quad \quad \quad \quad \quad \quad $\rhd_{tts} \quad x =_{v} w$}
	\BinaryInfC{$ x =_{\tau(s,\tau(\sigma(s),v))} w : D$}
\end{bprooftree}

\end{center}

\bigskip

\begin{center}
\begin{bprooftree}
	\AxiomC{$x =_{s} w : D$}
	\UnaryInfC{$w =_{\sigma(s)} x : D$}
	\AxiomC{$ x =_{s} w : D$}
	\AxiomC{$w =_{v} z : D$}
	\BinaryInfC{$x =_{\tau(s,v)} z : D$}
	\RightLabel{\quad \quad \quad \quad \quad \quad $\rhd_{ss} \quad w =_{v} z$}
	\BinaryInfC{$w =_{\tau(\sigma(s), \tau(s,v))} z : D$}
\end{bprooftree}

\end{center}

\bigskip

A full proof of confluence can be found in ~\cite{Anjo1,Ruy2,Ruy3,RuyAnjolinaLivro}.
 
\subsubsection{Normalization procedure}

We can now provide two normalization theorems:

\begin{theorem}[normalization ~\cite{RuyAnjolinaLivro}]
	Every derivation in the $\mathit{LND_{EQ}-TRS}$ converts to a normal form.
\end{theorem}

\begin{proof}
	Direct consequence of the termination property.
\end{proof}

\begin{theorem}[strong normalization ~\cite{RuyAnjolinaLivro}]
	Every derivation in the $\mathit{LND_{EQ}-TRS}$ converts to a unique normal form.
\end{theorem}

\begin{proof}
	Direct consequence of the termination and confluence properties.
\end{proof}

In this sense, every proof can be reduced to a normal one. To do so, one should identify the redundancies and, based on the rewrite rules, a proof can be constructed without any redundancies. We demonstrate this in the following example ~\cite{RuyAnjolinaLivro}:

\bigskip

\begin{center}
\begin{bprooftree}
	\AxiomC{$f(x,z) =_{s} f(w,y) : D$}
	\UnaryInfC{$f(w,y) =_{\sigma(s)} f(x,z) : D$}
	\AxiomC{$x =_{r} c : D$}
	\BinaryInfC{$f(w,y) =_{sub_{L}(\sigma(s), r)} f(c,z) : D$}
	\UnaryInfC{$f(c,z) =_{\sigma(sub_{L}(\sigma(s), r))} f(w,y) : D$}
	\AxiomC{$ y =_{t} b : D$}
	\BinaryInfC{$f(c,z) =_{sub_{L}(\sigma(sub_{L}(\sigma(s), r)))} f(w,b) : D$}
\end{bprooftree}

\end{center}

\bigskip

This deduction generates the following path: $sub_{L}(\sigma(sub_{L}(\sigma(s), r)))$. This path is not in normal form, having two redundancies ~\cite{RuyAnjolinaLivro}:
	
\begin{center}
	$sub_{L}(\sigma(sub_{L}(\sigma(s), r))) \rhd_{ssbl} sub_{L}(sub_{R}(\sigma(r),\sigma(\sigma(s)),t)$
	
	$sub_{L}(sub_{R}(\sigma(r),\sigma(\sigma(s)),t) \rhd_{ss} sub_{L}(sub_{R}(\sigma(r), s),t)$
\end{center}
	
Thus, we can identify those reductions and conceive a deduction without any redundancies ~\cite{RuyAnjolinaLivro}:

\bigskip

\begin{bprooftree}
	\AxiomC{$x =_{r} c : D$}
	\UnaryInfC{$c =_{\sigma(r)} x : D$}
	\AxiomC{$f(x,z) =_{s} f(w,y) : D$}
	\BinaryInfC{$f(c,z) =_{sub_{R}(\sigma(r),s)} f(w,y) : D$}
	\AxiomC{$y =_{t} b : D$}
	\BinaryInfC{$f(c,z) =_{sub_{L}(sub_{R}(\sigma(r), s),t)} f(w,b) : D$}
\end{bprooftree}

\bigskip

\subsection{Rewrite equality}

As we have just seen, the $\mathit{LND_{EQ}-TRS}$ has $39$ rewrite rules.  We call each rule  a {\em rewrite rule\/} (abbreviation: {\em rw-rule\/}). We provide the following definition:

\begin{definition}[Rewrite Rule ~\cite{Art1}]
 An $rw$-rule is any of the rules defined in $\mathit{LND_{EQ}-TRS}$.
\end{definition}

Similarly to the $\beta$-reduction of $\lambda$-calculus, we have a definition for rewrite reduction:

\begin{definition}[Rewrite reduction ~\cite{Art1}]
	Let $s$ and $t$ be computational paths. We say that $s \rhd_{1rw} t$ (read as: $s$ $rw$-contracts to $t$) iff we can obtain $t$ from $s$ by the application of only one $rw$-rule. If $s$ can be reduced to $t$ by finite number of $rw$-contractions, then we say that $s \rhd_{rw} t$ (read as $s$ $rw$-reduces to $t$).
	
\end{definition}

We shall also define rewrite contractions and equality:

\begin{definition}[Rewrite contraction and equality ~\cite{Art1}]
Let $s$ and $t$ be computational paths. We say that $s =_{rw} t$ (read as: $s$ is $rw$-equal to $t$) iff $t$ can be obtained from $s$ by a finite (perhaps empty) series of $rw$-contractions and reversed $rw$-contractions. In other words, $s =_{rw} t$ iff there exists a sequence $R_{0},....,R_{n}$, with $n \geq 0$, such that

\centering $(\forall i \leq n - 1) (R_{i}\rhd_{1rw} R_{i+1}$ or $R_{i+1} \rhd_{1rw} R_{i})$

\centering  $R_{0} \equiv s$, \quad $R_{n} \equiv t$
\end{definition}

A fundamental result is the fact that rewrite equality is an equivalence relation ~\cite{Art1}:
 
\begin{proposition}\label{proposition3.7} Rewrite equality is transitive, symmetric and reflexive.
\end{proposition}

\begin{proof}
	Stems directly from the fact that $rw$-equality is the transitive, reflexive and symmetric closure of $rw$.
\end{proof}

Rewrite reduction and equality play fundamental roles in the groupoid model of a type based on computational paths, as we shall see hereafter.

\subsection{$\mathit{\textbf{LND}_{\textbf{EQ}}\textbf{-TRS}}$(2)}

Until now, we have concluded in this subsection that there exist redundancies which are resolved by a system called $\mathit{LND_{EQ}-TRS}$. This system establishes rules that reduce these redundancies. Moreover, we have concluded that these redundancies are simply redundant uses of the equality axioms shown in \emph{section 2}. In fact, since these axioms only define an equality theory for type theory, we can be more specific and say that these are redundancies of the equality of type theory. As we have mentioned, the $\mathit{LND_{EQ}-TRS}$ has a total of $39$ rules~\cite{Anjo1, Ruy1}. 

Since the $rw$-equality is based on the rules of  $\mathit{LND_{EQ}-TRS}$, one can imagine the high number of redundancies that $rw$-equality could cause. In fact, a thorough study of all the redundancies caused by these rules led to the work done in ~\cite{Arttese}, which is solely interested in the redundancies caused by the fact that $rw$-equality is transitive, reflexive and symmetric with the addition of only one specific $rw_{2}$-rule. Thus, a system called $LND_{EQ}-TRS_{2}$ was created, which resolves all the redundancies caused by $rw$-equality (in the same way that $\mathit{LND_{EQ}-TRS}$ resolves all the redundancies caused by equality). 

Since we know that $rw$-equality is transitive, symmetric and reflexive, it should have the same redundancies that the equality had involving only those properties. Given that $rw$-equality is merely a sequence of $rw$-rules (which is also similar to equality, since equality is only a computational path, i.e., a sequence of identifiers), we could identify these sequences. Thus, if $s$ and $t$ are $rw$-equal because there exists a sequence $\theta: R_{0},....,R_{n}$ that justifies the $rw$-equality, then we can write that $s =_{rw_{\theta}} t$. Thus, by using $rw$-equality, we are able to rewrite all the rules which originated the ones involving $\tau$, $\sigma$ and $\rho$. For example, we have ~\cite{Art1}:

\bigskip

\begin{prooftree}
	\hskip - 155pt
	\AxiomC{$x =_{rw_{t}} y : A$}
	\AxiomC{$y =_{rw_{r}} w : A$}
	\BinaryInfC{$x =_{rw_{\tau(t,r)}} w : A$}
	\AxiomC{$w =_{rw_{s}} z : A$}
	\BinaryInfC{$x =_{rw_{\tau(\tau(t,r),s)}} z : A$}
\end{prooftree}

\begin{prooftree}
	\hskip 4cm
	\AxiomC{$x =_{rw_{t}} y : A$}
	\AxiomC{$y=_{rw_{r}} w : A$}
	\AxiomC{$w=_{rw_{s}} z : A$}
	\BinaryInfC{$y =_{rw_{\tau(r,s)}} z : A$}
	\LeftLabel{$\rhd_{tt_{2}}$}
	\BinaryInfC{$x =_{rw_{\tau(t,\tau(r,s))}} z : A$}
\end{prooftree}

\bigskip

Therefore, we obtain the rule $tt_{2}$ which resolves one of the redundancies caused by the transitivity of $rw$-equality (the $2$ in $tt_{2}$ indicates that it is a rule that resolves a redundancy of $rw$-equality). In fact, using the same reasoning, we can obtain, for $rw$-equality, all the redundancies that we have shown in Definition \ref{LND-TRS}. In other words, we have $tr_{2}$, $tsr_{2}$, $trr_{2}$, $tlr_{2}$, $sr_{2}$, $ss_{2}$ and $tt_{2}$. Since we now are provided with rules of $LND_{EQ}-TRS_{2}$, we can use all the concepts that we have just defined for $\mathit{LND_{EQ}-TRS}$. The only difference is that instead of having $rw$-rules and $rw$-equality, we have $rw_{2}$-rules and $rw_{2}$-equality.

There is an important rule specific to this system. It stems from the fact that transitivity of reducible paths can be reduced in different ways, but generating the same result. For example, consider the simple case of $\tau(s,t)$ and consider that it is possible to reduce $s$ to $s'$ and $t$ to $t'$. There are two possible $rw$-sequences which reduce this case: The first one is $\theta: \tau(s,t) \rhd_{1rw} \tau(s',t) \rhd_{1rw} \tau(s',t')$ and the second $\theta': \tau(s,t) \rhd_{1rw} \tau(s,t') \rhd_{1rw} \tau(s',t')$. Both $rw$-sequences obtained the same result in similar ways, the only difference being the choices that have been made at each step. Since the variables, when considered individually,  followed the same reductions, these $rw$-sequences should be considered redundant relative to each other and, for that reason, there should be an $rw_{2}$-rule that establishes this reduction. This rule is called \emph{independence of choice} and is denoted by $cd_{2}$. Since we already understand the necessity of such a rule, we can define it formally:

\begin{definition}[Independence of choice ~\cite{Art1}]
	Let $\theta$ and $\phi$ be $rw$-equalities expressed by two $rw$-sequences: $\theta: \theta_{1},...,\theta_{n}$, with $n \geq 1$, and $\phi: \phi_{1},...,\phi_{m}$, with $m \geq 1$. Let $T$ be the set of all possible $rw$-equalities from $\tau(\theta_{1},\phi_{1})$ to $\tau(\theta_{n},\theta_{m})$ described by the following process: $t \in T$ is of the form  $\tau(\theta_{l_{1}},\phi_{r_{1}}) \rhd_{1rw} \tau(\theta_{l_{2}},\phi_{r_{2}}) \rhd_{1rw} ... \rhd_{1rw} \tau(\theta_{l_{x}},\phi_{r_{y}})$, with $l_{1} = 1, r_{1} = 1$,  $l_{x} = n, r_{y} = m$ and $l_{i + 1} = 1 + l_{i}$ and $r_{i + 1} = r_{i}$ or $l_{i + 1} = l_{i}$ and $r_{i + 1} = 1 + r_{i}$. The independence of choice, denoted by $cd_{2}$, is defined as the rule of $LND_{EQ}-TRS_{2}$ that establishes the equality between any two different terms of $T$. In other words, if $x,y \in T$ and $x \neq y$, then $x =_{cd_{2}} y$ and $y =_{cd_{2}} x$.
	
\end{definition}

Analogously to the $rw$-equality, $rw_{2}$-equality is also an equivalence relation ~\cite{Art1}:

\begin{proposition}
	$rw_{2}$-equality is transitive, symmetric and reflexive.
\end{proposition}

\begin{proof}
	Analogous to Proposition \ref{proposition3.7}.
\end{proof}

\section{A topological application  of labelled natural deduction}	
Once we have built up all the necessary bases of computational paths to develop our work, it would be interesting to consult two proofs of the calculation of the fundamental group of the circle: The first is the mathematical proof that appears in the book of algebraic topology ~\cite {Munkres} in chapter $9$, section $54$. The second is a proof using homotopic type theory, which is in the book in ~\cite{hott} in chapter 8. Both cases provide the proofs of the fundamental group of the circle, but in order to obtain such success the amount of information needed is much higher and much more complex than we will propose in the next chapters. 

In homotopy theory, the fundamental group is the one formed by all equivalence classes up to homotopy of paths (loops) starting from a point $x_0$ and also ending at $x_0$. Since we use computational paths as the syntactic counterpart of homotopic paths in type theory, we will use computational paths to propose some definitions that will be seen below.

Before we begin these definitions, we are going to formalize some notations that will be recurrent in the text. Consider the type $S^{1}$ (Circle), and let $x_0:S^{1} $ be a base term of the type, and $x_{0}\underset{\alpha}{=}x_{0}$ be a computational path that starts and ends at the term $x_0 $, going around the circle clockwise. We can then define this path as a base path, capable of generating any path in the circle, and denote it by $loop_{x_0}$. However, for simplicity, we will omit the $x_0$, but it is implied that our loops will be made at the base point and we will denote these loops by $loop^1$. 

Thus, the path (loop) formed by two turns based on $x_0$, around the circle in a clockwise direction, can be denoted by $loop^2$; a counterclockwise loop for $loop^{- 1} $, in general, $loop^{n}$ denotes the path formed by $n$ clockwise turns in the circle, based on $ x_0 $, with $n\in\mathbb{Z}$ . Particularly, if $n =0$, we can say that this is the homotopic path to the point and denotes it by $loop^{0}$. 

Now, imagine the path formed by three clockwise turns and two counterclockwise turns. This path is different from $loop^1$, but it is equivalent to it, that is, we can say that it is a rewrite of the  computational path $loop$ or $loop^{1}$, so it is relevant here to define a rewrite equivalence, and we can simply denote for $[loop^n]_{rw}$ every computational path that is equivalent, or a rewrite, of the $loop^{n}$ path. Now, we can proceed with the following definitions: 

\begin{definition} Let
\begin{itemize}
    \item [(i)] $A$ be a type.
    \item [(ii)] $x_0:A$ a base point.
    \item [(iii)]  $x_{0}\underset{\alpha_i}{=}x_{0}$, be a family of generator paths with $i\in I$.
    \item [(iv)] A family of relationships between the terms paths $\tau_{j}(x_0\underset{\alpha_r}{=}x_0,x_0\underset{\alpha_s}{=}x_0)$.
    \end{itemize}
    We can define the structure $\Pi_{1}(A, x_0)$ as the set of terms $\alpha_{x_0}$, given by finite applications of $\tau$, $\sigma$, and $\rho$ in $\alpha_i$, modulo $rw$ equality and modulo family of identity type terms $Id_{\tau_j}$.

 \label{defestruturageral}
\end{definition}

Since each element in $\Pi_{1}(A, x_0)$ is a loop in $x_0$, we shall give an important definition indispensable to our work:

\begin{definition}
	 We can define and denote by 
	 $$[loop^n]_{rw}$$
the path naturally obtained by the application of the path-axioms $\rho$, $\tau$ and $\sigma$ to the base path  $x_0 \underset{loop}{=} x_0$, where $n\in \mathbb{N}.$ Particularly we can say:
	 
	 \begin{itemize}
	     \item [(i)] $[loop^0]_{rw}=[\rho_{x_0}]_{rw}$,  $n=0$.
	     \item [(ii)]
	     $[loop^1]_{rw}=[loop]_{rw}$
	     \item [(iii)] $[loop^{n}]_{rw}=\tau\big([loop^{n-1}]_{rw},[loop^1]_{rw}\big)$, $n>0$.
	     \item [(iv)] $[loop^n]_{rw} = \sigma([loop^{-n}]_{rw})$, $-n>0$.
	 \end{itemize}
	\end{definition}

	For example, we have: 
	\begin{itemize}
	    \item  [a)]$\tau([loop^1]_{rw}, [loop^1]_{rw})= [loop^{2}]_{rw}$
	    \item  [a)]$\tau \Big(\sigma([loop^1]_{rw}), \sigma([loop^1]_{rw})\Big)=\sigma([loop^2]_{rw})=  [loop^{-2}]_{rw}$
	    \item [c)] $\tau \Big(\sigma([loop^1]_{rw}), [loop^1]_{rw}\Big)\underset{tsr}{=} [\rho]_{rw}$.
	 
	\end{itemize}

Here we need to provide relevant information regarding the equalities we can obtain using these paths. Consider the following examples:

\small
	\begin{itemize}
	    \item [($p_1$)]	\begin{eqnarray*}
            \tau \bigg( \tau\big( [loop^1]_{rw},[loop^1]_{rw}\big),\sigma \big([loop^1]_{rw}\big)\bigg) &\underset{tt}{=}&  \tau \bigg( [loop^1]_{rw}, \tau\Big([loop^1]_{rw},\sigma \big([loop^1]_{rw}\big)\Big)\bigg) \\
             &\underset{tr}{=}&  \tau\big( [loop^1]_{rw},[\rho]_{rw}\big)  \\
             &\underset{trr}{=}& [loop^1]_{rw} \\
             \end{eqnarray*}

	    \item  [($p_2$)]\begin{eqnarray*}
            \tau \bigg( \tau\Big( [loop^1]_{rw},\sigma \big( [loop^1]_{rw}\big)\Big),[loop^1]_{rw}\bigg) &\underset{tr}{=}&  \tau\big( [\rho]_{rw},[loop^1]_{rw}\big)  \\
             &\underset{tlr}{=}& [loop^1]_{rw} \\
             \end{eqnarray*} 
	    
	\end{itemize}

Notice that the paths $(p_1)$ and $(p_2)$ initially appear to be distinct paths. However, by only applying the properties of computational paths, together with the rewrite rules (\textit{rw-rules}), we end up with the path $[loop^1]_{rw}$ in both derivations. So we can affirm that:

By $(p_1)$,  $$\tau \bigg( \tau\big( [loop^1]_{rw},[loop^1]_{rw}\big),\sigma \big([loop^1]_{rw}\big)\bigg)\underset{trr}{=} [loop^1]_{rw}$$ and by $(p_2)$,  $$\tau \bigg( \tau\Big( [loop^1]_{rw},\sigma \big( [loop^1]_{rw}\big)\Big),[loop^1]_{rw}\bigg)\underset{tlr}{=} [loop^1]_{rw}.$$

They are said to be \textit{rw-equal} to the base path  $[loop^1]_{rw}$ because they can be rewritten to $[loop^1]_{rw}$ after the \textit{rw-rules} are applied. Therefore, it can be said that these paths are in the same equivalence class as $[loop^1]_{rw}$ and thus, they are equal up to \textit{rw}-equality.

\subsection {Fundamental group of the circle}

	\begin{definition}[The circle $S^1$]
		The circle is the type generated by:
		
		\begin{itemize}
			\item [(i)] A base point - $x_0 : S^1$	
			\item [(ii)] A base computational path - $x_0 \underset{loop}{=} x_0 : S^1$.
		\end{itemize}
	\end{definition}
	
	The first thing one should notice is that this definition does not use only the points of the type $S^1$, but also a base computational path  called $loop$ between those points. That is why it is called a higher inductive type ~\cite{hott}. Our approach differs from the one developed in the HTT book ~\cite{hott} in the fact that we do not need to simulate the path-space between those points, since we add computational paths to the syntax of the theory. 
	
	In Martin-L\"of's type theory, the existence of those additional paths emerges from establishing that the paths should be freely generated by the constructors ~\cite{hott}. In our approach, we do not have to appeal to this kind of argument, since all paths naturally emerge from direct applications of the axioms and the inference rules which define the theory of equality. We proceed with the following definition:

\begin{definition} In $S^1$,  we define the following canonical loops (canonical paths):

 \begin{itemize}
	     \item [(i)] $[loop^0]_{rw}=[\rho_{x_0}]_{rw}$, $n=0$
	     \item [(ii)]
	     $[loop^1]_{rw}=[loop]_{rw}$, $n=1$.
	     \item [(iii)] $[loop^{n}]_{rw}=\sigma([loop^{-n}]_{rw})$, $n<0$. 
	     
	     \item [(iv)] $[loop^{n}]_{rw}=\tau\big([loop^{n-1}]_{rw},[loop^1]_{rw}\big)$, $n>0$. 
	 \end{itemize}
	 
\end{definition}

	\begin{lemma}\label{S1=loopn}
		All paths in $S^1$  are $rw$-equal to a path  $[loop^n]_{rw}$, for some $n \in \mathbb N$.
	\end{lemma}

	\begin{proof}
	Let $\varphi$ be a computational path in $S^1$.
	
	\begin{itemize}
	    \item [\textbf{I.}] If $\varphi=\rho$:
	        \begin{itemize}
	            \item [(i)] $\varphi=[loop^0]_{rw}$, $n=0$.
	            
	            \item [(ii)] $\varphi=\sigma([loop^n]_{rw})=\sigma(\sigma([loop^{-n}]_{rw}))\underset{ss}{=}[loop^{-n}]_{rw}=\rho$, $n=0.$
	      	    \item [(iii)] $\varphi=\tau\big([loop^m]_{rw}, [loop^n]_{rw}\big)=\rho$, if $m+n=0$. Therefore,
	            
	            \begin{eqnarray*}
	            \varphi=\tau\big([loop^m]_{rw}, [loop^n]_{rw}\big)&&=\tau\big([loop^{-n}]_{rw}, [loop^n]_{rw}\big)\\
	            &&=\tau\big([loop^{-n}]_{rw}, \sigma([loop^{-n}]_{rw})\big)\\
	            &&\underset{tr}{=}\rho.
	            \end{eqnarray*}
	      \end{itemize}
	    
	    \item [\textbf{II.}.] If $\varphi=\sigma([loop^{n}]):$
	    
	    \begin{itemize}
	        \item [(i)] For $n=0$ we have $\varphi=\sigma([loop^0]_{rw})=\rho$.
	        
	        \item [(ii)] Suppose true for $n=k$ that every path in $S^1$  is $rw$-equal to a path  $[loop^n]_{rw}$ . For $n=k+1$ we have:
	        
	        \begin{eqnarray*}
	            \varphi&&=\sigma([loop^{k+1}]_{rw})\\
	            &&=\sigma\Big(\tau\big([loop^k]_{rw}, [loop^1]_{rw}\big)\Big)\\
	            &&\underset{stss}{=}\tau\big(\sigma([loop^{k}]_{rw}),\sigma([loop^1]_{rw})\big)\\
	            &&=\tau\big([loop^{-k}]_{rw},[loop^{-1}]_{rw}\big)\\
	            &&\underset{}{=}[loop^{-k-1}]_{rw}\\
	            &&\underset{}{=}[loop^{-(k+1)}]_{rw}.
	            \end{eqnarray*}
	        
	    \end{itemize}
	    
	     \item [\textbf{III.}] If $\varphi=\tau\big([loop^{n-1}]_{rw},[loop^{1}]_{rw}\big):$
	     
	     \begin{itemize}
	         \item [(i)] For $n=0$, we have:
	         
	         $$\varphi=\tau([loop^{-1}]_{rw},[loop^{1}]_{rw})=\tau(\sigma([loop^{1}]_{rw}),[loop^{1}]_{rw})\underset{tsr}{=}\rho=[loop^{0}]_{rw}.$$
	         
	          \item [(ii)] Suppose true for $n=k$, to $n=k+1$ we have:

	     \end{itemize}
	      \begin{eqnarray*}
	            \varphi&&=\tau\big([loop^{k+1-1}]_{rw},[loop^{1}]_{rw}\big)\\
	            &&=\tau\Big([loop^{k}]_{rw},\tau([loop^{1}]_{rw})\Big)\\
	            &&\overset{hip}{=}\tau\Big(\tau([loop^{k-1}]_{rw},[loop^{1}]),[loop^{1}]_{rw}\Big)\\
	            &&\underset{rw}{=}[loop^{1}]_{rw}\circ[loop^{k}]_{rw}\\
	            &&=[loop^{k+1}]_{rw}.
	            \end{eqnarray*}
	\end{itemize}
	
	\end{proof}
	
		All paths in $S^1$  are $rw$-equal to a path  $[loop^n]_{rw}$, for some $n \in \mathbb N$.

	\begin{lemma}
	All paths $[loop^n]_{rw}$ in $S^1$ can be expressed in terms of $\rho$,$\tau$, $\sigma$ and their applications, starting from the base path  $[loop^1]_{rw}$.
	\end{lemma}

	\begin{proof}
		For the base case $[\rho]_{rw}$, it is trivially true, since we define it as being equal to $[loop^0]_{rw}$. From $[\rho]_{rw}$, one can construct more complex paths by composing with $[loop^1]_{rw}$ or $\sigma([loop^1]_{rw})$ at each step.
		Concatenating the paths we obtain:
		\begin{itemize}
			\item [(i)] A path of the form   $[\rho]_{rw}$ concatenated with $[loop^1]_{rw}$:   $$[\rho]_{rw} \circ[loop^1]_{rw} = \tau([loop^1]_{rw},[\rho]_{rw}) \underset{trr}{=} [loop^1]_{rw}.$$ 
			\item [(ii)] A path of the form $[\rho]_{rw}$ concatenated with $\sigma([loop^1]_{rw})$:  $$[\rho]_{rw} \circ \sigma([loop^1]_{rw}) = \tau(\sigma([loop^1]_{rw}),[\rho]_{rw}) \underset{trr}{=}  \sigma([loop^1]_{rw}) = [loop^{-1}]_{rw}.$$
		
			\item [(iii)] A path of the form $[loop^n]_{rw}$ concatenated with $[loop^1]_{rw}$:  $$[loop^n]_{rw} \circ [loop^1]_{rw} = \tau([loop^1]_{rw}, [loop^n]_{rw}) = [loop^{n+1}]_{rw}.$$
		
			\item [(iv)] A path of the form  $[loop^n]_{rw}$ concatenated with $\sigma([loop^1]_{rw})$:
			\begin{eqnarray*}
			[loop^n]_{rw} \circ \sigma([loop^1]_{rw})&=&  \tau\big(\sigma([loop^1]_{rw}),[loop^n]_{rw}\big)\\
			&=& \tau\bigg(\sigma([loop^1]_{rw}),\tau\Big([loop^{1}]_{rw},[loop^{n-1}]_{rw}\Big)\bigg)\\
			&\underset{\sigma(tt)}{=}&\tau\bigg(\tau\Big(\sigma([loop^1]_{rw}),[loop^{1}]_{rw}\Big),[loop^{n-1}]_{rw}\bigg)\\
			&\underset{tsr}{=}&\tau\Big([\rho]_{rw},[loop^{n-1}]_{rw}\Big)\\
			&\underset{tlr}{=}& [loop^{n-1}]_{rw}.\\
			\end{eqnarray*}
		
		\item [(v)] A path of the form  $[loop^{-n}]_{rw}$ concatenated with $[loop^1]_{rw}$:
			\begin{eqnarray*}
			[loop^{-n}]_{rw} \circ [loop^1]_{rw}&=&  \tau\big([loop^1]_{rw},[loop^{-n}]_{rw}\big)\\
			&=& \tau\bigg([loop^1]_{rw},\tau\Big(\sigma([loop^{1}]_{rw}),[loop^{-(n-1)}]_{rw}\Big)\bigg)\\
			&\underset{\sigma(tt)}{=}&\tau\bigg(\tau\Big([loop^{1}]_{rw},\sigma([loop^1]_{rw})\Big),[loop^{-(n-1)}]_{rw}\bigg)\\
			&\underset{tr}{=}&\tau\Big([\rho]_{rw},[loop^{-(n-1)}]_{rw}\Big)\\
			&\underset{tlr}{=}& [loop^{-(n-1)}]_{rw}.\\
			\end{eqnarray*}
			
			\item [(vi)] a path of the form   $[loop^{-n}]_{rw}$ concatenated with $\sigma([loop^{1}]_{rw})$:
			
			\begin{eqnarray*}
			[loop^{-n}]_{rw} \circ \sigma([loop^{1}]_{rw})&=&  \tau\big(\sigma([loop^{1}]_{rw}),[loop^{-n}]_{rw}\big)\\
			&=& \tau\bigg(\sigma([loop^{1}]_{rw}),\tau\Big([loop^{1}]_{rw},[loop^{-(n+1)}]_{rw}\Big)\bigg)\\
			&\underset{\sigma(tt)}{=}&\tau\bigg(\tau\Big(\sigma([loop^1]_{rw}),[loop^{1}]_{rw}\Big),[loop^{-(n+1)}]_{rw}\bigg)\\
			&\underset{tsr}{=}&\tau\Big([\rho]_{rw},[loop^{-(n+1)}]_{rw}\Big)\\
			&\underset{tlr}{=}& [loop^{-(n+1)}]_{rw}.\\
			\end{eqnarray*}

		\end{itemize}

	\end{proof}

For simplicity, we will denote by  $x_0 \underset{r}{=} x_0$  whenever we refer to a computational path $r$ generated by $\rho,\sigma$ and $\tau$.
	\begin{proposition}
	 $\Pi_{1}(S^1,x_0)$ provided with operations $\rho, \sigma, \tau$ is a group.
	\end{proposition}
	
	\begin{proof}
		 Given any $x_0 \underset{r}{=} x_0 : S^1$ and $x_0 \underset{t}{=} x_0 : S^1$, we need to check the group conditions:
		
		\begin{itemize}
			
			\item [(i)] \textbf{Closure:} Given $x_0 \underset{r}{=} x_0 : S^1$ and $x_0 \underset{s}{=} x_0 : S^1$, $r \circ s$ must be a member of the group. Indeed, $r \circ s = \tau(s,r)$ is a computational path $x_0 \underset{\tau(s,r)}{=} x_0 : S^1$.
		
			\item [(ii)]\textbf{Inverse:} Every member of the group must have an inverse. Indeed, if we have a path $r$, we can apply $\sigma(r)$. We claim that $\sigma(r)$ is the inverse of $r$, since:
			
			\begin{center}
				$\sigma(r) \circ r = \tau(r, \sigma(r)) \underset{tr}{=} \rho$
				
				$r \circ \sigma(r) = \tau(\sigma(r), r) \underset{tsr}{=} \rho$
			\end{center}

			Since we are working up to $rw$-equality, the equalities hold strictly.
			
			\item [(iii)] \textbf{Identity:} We use the path $x_0 \underset{\rho}{=} x_0 : S^1$ as the identity. Indeed, we have:
			
			\begin{center}
				$r \circ \rho = \tau(\rho,r) \underset{tlr}{=} r$
				
				$\rho \circ r = \tau(r,\rho) \underset{trr}{=} r$.
			\end{center}
			
			\item [(iv)]\textbf{Associativity:} Given any member of the group $x_0 \underset{r}{=}x_0 : S^1$, $x_0 \underset{t}{=} x_0$ and $x_0 \underset{s}{=} x_0$, we want that $r \circ (s \circ t) = (r \circ s) \circ t$:
			
			\begin{center}
				$r \circ (s \circ t) = \tau(\tau(t,s), r) \underset{tt}{=} \tau(t,\tau(s,r)) = (r \circ s) \circ t$
			\end{center}
		\end{itemize}
		
		Thus, all conditions have been satisfied. 
	\end{proof}
	
	Therefore, the structure $\Big(\Pi_{1}(S^1,x_0), \rho, \sigma, \tau \Big)$ is indeed a group. We will, for simplicity, denote it in what follows for $\Pi_{1}(S^1,x_0)$, and we will call it The Fundamental Group of $S^1$.
	
	In ~\cite{hott}, the next theorem was proved by defining a pair of encode and decode functions. There it was necessary to simulate a path-space, and in the end the work was very laborious. Nevertheless, since our computational paths are already part of the syntax, there is no need to rely on this kind of approach to simulate a path-space. In ~\cite{Munkres} the proof of this theorem is quite laborious. By working directly with the concept of computational paths, we hope that these same calculations can be performed more simply and in such a way that is accessible to more readers. 

	\begin{theorem}
		$\Pi_{1}(S,x_0) \simeq \mathbb{Z}$
	\end{theorem}

\begin{proof}

Consider the application defined and denoted by: 

\begin{eqnarray*}
toPath: &&\mathbb Z \rightarrow \Pi_{1}(S)\\
&&z \rightarrow toPath(z)= [loop^{z}]_{rw}.
\end{eqnarray*}

\begin{itemize}
    \item [(i)] $toPath$ is a homomorphism.
    
    Let $z=n+m \in \mathbb{Z}$, then:
   \begin{eqnarray*}
         toPath(z)&=&toPath(n+m)\\
         &=&[loop^{n+m}]_{rw}\\
         &=&\tau([loop^{n}]_{rw},[loop^{m}]_{rw})\\
         &=&toPath(m)\circ toPath(n).
  \end{eqnarray*}  
  
  On the other hand, as $z=m+n$ we have:
  
  \begin{eqnarray*}
         toPath(z)&=&toPath(m+n)\\
         &=&[loop^{m+n}]_{rw}\\
         &=&\tau([loop^{m}]_{rw},[loop^{n}]_{rw})\\
         &=&toPath(n)\circ toPath(m).
  \end{eqnarray*}  
  
  Thus, $toPath(n+m)=toPath(n)\circ toPath(m)$.
  
  \item [(ii)] $toPath$ is surjective.

By \textbf{Lemma \ref{S1=loopn}}, as every path in $S^1$ is $rw$-equal to a path $[loop^{i}]_{rw}$, we have that for all paths $[loop^{i}]_{rw} \in \Pi_{1}(S^1), \exists i \in \mathbb{Z}$, such that, $ toPath(i)=[loop^{i}]_{rw}.$

   \item [(iii)] $Ker(toPath)=\{0\}$.
   
   Suppose there is $z\neq 0 \in \mathbb{Z}$, such that $z \in Ker(toPath)$. Thus,
   
   $$toPath(z)=toPath(z+0)\overset{hom}{=}toPath(z)\circ toPath(0)=\tau(\rho,[loop^{z}]_{rw})\overset{Ker}{=}\rho.$$
   
   If $\tau(\rho,\alpha)=\rho$, by $rw$-rule $\underset{tr}{\triangleright}$ we have, $\alpha= \sigma(\rho) \Rightarrow z=0  \rightarrow \leftarrow$. Therefore,  $$Ker(toPath)=\{0\}.$$
\end{itemize}

As $ToPath$ is a surjective homomorphism with $Ker(toPath)=\{0\}$, then $toPath$ is an isomorphism, that is, $\Pi_{1}(S,x_0) \simeq \mathbb{Z}$.
\end{proof}

\subsection{Fundamental group of the torus}
	
Consider $\mathbb{T}^{2}$ as the surface known as Torus and the point $x_{0}\in \mathbb{T}^{2}$. We will prove using computational paths that the fundamental group of the torus is isomorphic to $\mathbb{Z} \times \mathbb{Z}$. Here we will also use Definition \ref{defestruturageral} with some simple adaptations. We will continue to work with paths up to $rw$-equality.

\begin{figure}[!htb]
\centering
\includegraphics[width=0.4\columnwidth]{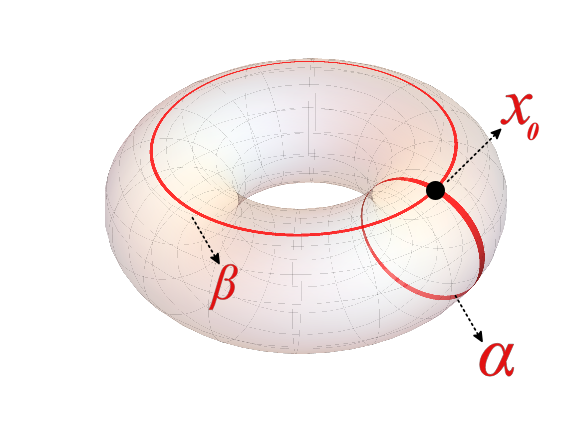}
\caption{Paths $\alpha$ and $\beta$ with base point $x_{0}$ in Torus} 
\label{fig1}
\end{figure}

Since the fundamental groups are obtained by analysing the loops, we will be interested in working with \textit{loops} that cannot be homotopic to base point $x_{0}$, like \textit{loops} $\alpha$ and $\beta$. These loops will be the generators of $\mathbb{T}^{2}$, as shown in\textbf{ Figure\ref{fig1}}, so we can give them a special definition for both. For simplicity, we will make two brief definitions to distinguish how the $\alpha$ $ \beta$ loops traverse the surface of Figure\ref{fig1}, though this does not alter the formal definition of the loops in any way. These definitions aim at making the evidence visually clearer for the reader, avoiding very visually cluttered expressions, thus enabling a better understanding.

\begin{definition}[{vertical loop}]
   We define and denote by 
   
    $$\alpha^n=[loop^{n}_{v}]_{rw}$$
    
    the path that passes through the inner part of $\mathbb{T}^{2}$ in the vertical direction, naturally obtained by applications of the path-axioms $\rho$, $\tau$ and $\sigma$ to the base path  $x_0 \underset{\alpha}{=} x_0$, where $n\in \mathbb{Z}.$ Particularly, we have
	 
	 \begin{itemize}
	     \item [(i)] $[loop^{0}_v]_{rw}=[\rho]_{rw}=\alpha^0$, $n=0$.
	     \item [(ii)] $[loop^{n+1}_v]_{rw}=\tau\big([loop^n_v]_{rw},[loop^1_v]_{rw}\big)=\alpha^{n+1}$, $n>0$.
	     \item [(iii)] $[loop^{n}_v]_{rw}=\sigma([loop^{-n}_v]_{rw}) = \alpha^{-n}$, $n<0$.
	 \end{itemize}
	 \label{defloopv}
	\end{definition}
   
    In \textbf{Figure \ref{fig1}}, this vertical path (loop) has the same orientation of the path denoted by $\alpha$.

\begin{definition}[{horizontal loop}] We define and denote by 
   
    $$\beta^m=[loop^{m}_{h}]_{rw}$$

   the path that passes through the inner part of  $\mathbb{T}^{2}$ in the horizontal direction, naturally obtained by applications of the path-axioms $\rho$, $\tau$ and $\sigma$ to the base path  $x_0 \underset{\beta}{=} x_0$, where $n\in \mathbb{Z}.$ Particularly, we have:
	 
	 \begin{itemize}
	    \item [(i)] $[loop^{0}_h]_{rw}=[\rho]_{rw}=\beta^0$, $m=0$.
	     \item [(ii)] $[loop^{m+1}_h]_{rw}=\tau\big([loop^m_h]_{rw},[loop^1_h]_{rw}\big)=\beta^{m+1}$, $m>0$.
	     \item [(iii)] $[loop^{m}_h]_{rw}=\sigma([loop^{-m}_h]_{rw}) = \beta^{-m}$, $m<0$.
	 \end{itemize}
\label{deflooph}
	\end{definition}
   
   In \textbf{Figure \ref{fig1}}, this horizontal path (loop) has the same orientation of the path denoted by $\beta$.  By Definitions \ref{defloopv} and \ref{deflooph}, we can also represent the path homotopic to the constant one by: $[\rho]_{rw}=\alpha^0\beta^0$, or $[\rho]_{rw}=\alpha^0$, or $[\rho]_{rw}=\beta^0$. For simplicity, we denote it by $\rho$.

We now give the formal definition of the torus in homotopy type theory:

\begin{definition} \label{torodef}
 The torus $\mathbb{T}^2$ is generated by:
 
 \item[ (i)] A base point  $x_0:\mathbb{T}^2$.
 \item[(ii)] Two base paths $\alpha$ and $\beta$ such that: $x_{0}\underset{\alpha}{=} x_{0}$ \hspace{0.2cm} and \hspace{0.2cm}  $x_{0}\underset{\beta}{=} x_{0}$.
 \item[(iii)] One path $co$ that establishes $\beta\alpha\underset{co}{=}\alpha\beta$, i.e., a term $co:Id(\beta\alpha,\alpha\beta)$.
\end{definition}

Based on definition \ref{torodef}, we can establish the following definition in computational paths:

\begin{definition} In $\mathbb{T}^2$,  we define the following canonical loops (canonical paths):

 \begin{itemize}
	     \item [(i)] A base point  $\alpha^0\beta^0\underset{rw}{=}[\rho_{x_0}]_{rw}.$
	     \item [(ii)] The path $\beta^{m}\alpha^{n}=\tau(\alpha^{n},\beta^{m})$.
	     
	     \item [(iii)] The path $\sigma(\beta^{m}\alpha^{n})=\sigma(\tau(\alpha^{n},\beta^{m}))$.
	     
	     \item [(iv)] One path $co$ that establishes $\tau(\alpha^n,\beta^m)\underset{co}{=} \tau(\beta^m,\alpha^n).$
	     
	 \end{itemize}
\end{definition}

By ~\cite{Munkres}, given a point $x_{0}\in \mathbb{T}^2$, the Torus can be expressed as the quotient of a square whose sides are the base paths (loops) $\alpha$ and $\beta$, as shown in \textbf{Figure \ref{fig2}}.

\begin{figure}[H]
\centering
\includegraphics[width=0.4\columnwidth]{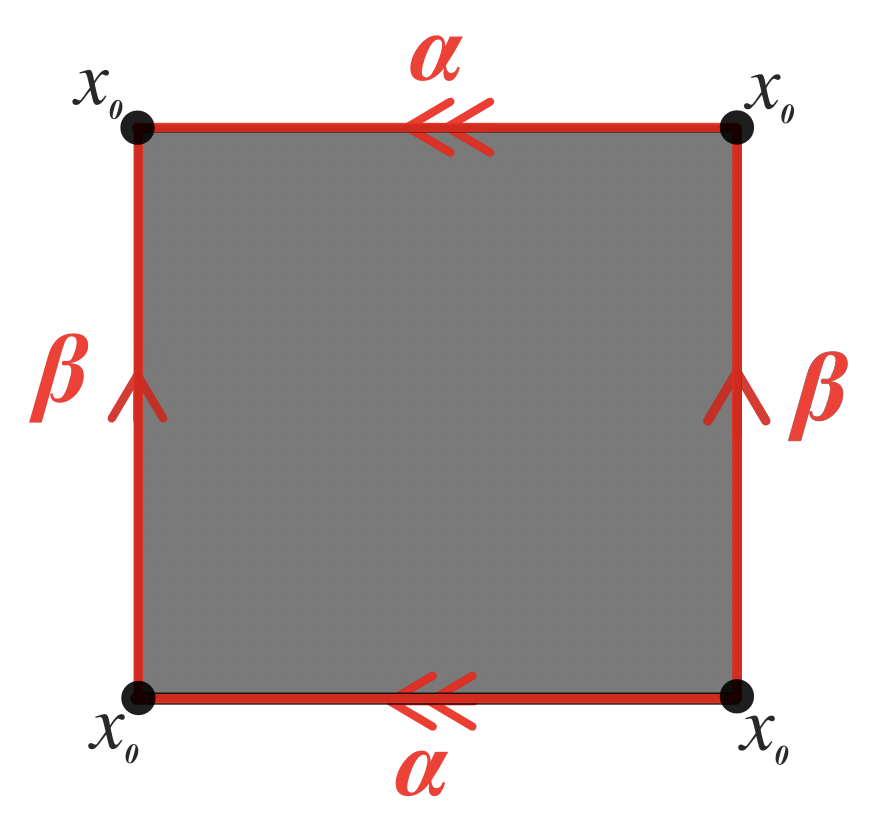}
\caption{Square Torus representation with oriented paths $\alpha$ and $\beta$ } 
\label{fig2}
\end{figure}

Consider the following path in the figure: $$\alpha^{-1}\circ \beta^{-1} \circ \alpha \circ \beta = \tau\bigg( \tau\Big(\tau(\beta,\alpha),\sigma(\beta)\Big),\sigma(\alpha)\bigg).$$

\begin{proposition}
The aforementioned path is $rw$-equal to the reflexive path.
\end{proposition}

\begin{proof}
Indeed,

\begin{eqnarray*}
\alpha^{-1}\circ \beta^{-1} \circ \alpha \circ \beta &=& \tau\bigg( \tau\Big(\tau(\beta,\alpha),\sigma(\beta)\Big),\sigma(\alpha)\bigg)\\
	&\underset{tt}{=}& \tau\bigg( \tau\Big(\beta,\tau \big(\alpha,\sigma(\beta)\big)\Big),\sigma(\alpha)\bigg)\\
	&\underset{co}{=}&  \tau\bigg( \tau\Big(\beta,\tau \big(\sigma(\beta),\alpha\big)\Big),\sigma(\alpha)\bigg)\\
    &\underset{\sigma(tt)}{=}&\tau\bigg( \tau\Big(\tau(\beta,\sigma(\beta)),\alpha\Big),\sigma(\alpha)\bigg)\\
    &\underset{tr}{=}&\tau\bigg( \tau\Big(\rho,\alpha\Big),\sigma(\alpha)\bigg)\\
    &\underset{tlr}{=}&\tau\Big( \alpha, \sigma(\alpha) \Big)\\
    &\underset{tr}{=}&\rho.\\
\end{eqnarray*}
and thus,  $$\alpha^{-1}\circ \beta^{-1} \circ \alpha \circ \beta =\tau\bigg( \tau\Big(\tau(\beta,\alpha),\sigma(\beta)\Big),\sigma(\alpha)\bigg) \underset{rw}{=} [\rho]_{rw}.$$
\end{proof}

\begin{lemma}\label{pathloopMN}
 All paths in $\mathbb{T}^2$ are    \textit{rw-equal} to the path $\beta^{m}\alpha^{n}$, with $m,n \in \mathbb{Z}$.

\begin{proof} Let $\varphi$ be a computational path in $\mathbb{T}^2$.

\begin{itemize}
    \item If $\varphi=\rho$ then $\varphi=\tau(\alpha^0,\beta^0)=\beta^{0}\alpha^{0}$.
    
    \item If $\varphi=\sigma(\mu)=\sigma(\beta^{m}\alpha^{n})=\sigma\Big(\tau(\alpha^{n},\beta^{m})\Big)\underset{stss}{=}\tau\Big(\sigma(\beta^{m}),\sigma(\alpha^{n})\Big)=\tau(\alpha^{-n},\beta^{-m})=\beta^{-m}\alpha^{-n}$.
    
    \item If 
    
    \begin{eqnarray*}
    \varphi &&= \tau(\mu,\omega)\\
    &&=\tau(\beta^{m}\alpha^{n},\beta^{r}\alpha^{s})\\
    &&=\tau\Big(\tau(\alpha
    ^{n},\beta^{m}),\tau(\alpha^{s},\beta^{r})\Big)\\
    &&=(\beta^{r}\circ\alpha
    ^{s}) \circ(\beta^{m}\circ\alpha^{n})\\
    &&=\beta^{r}\circ\alpha
    ^{s} \circ\beta^{m}\circ\alpha^{n}\\
    &&\underset{co}{=}\beta^{r}\circ\beta^{m}\circ \alpha
    ^{s}\circ\alpha^{n}\\
    &&=\beta^{m+r}\alpha
    ^{n+s}\\
    \end{eqnarray*}
    
\end{itemize}
\end{proof}

\end{lemma}

 \begin{lemma}

All paths $\beta^{m}\alpha^{n}$ in $\mathbb{T}^{2}$  can be expressed in terms of $\rho$,$\tau$, $\sigma$ and their applications, starting from the base paths   $\alpha^1=[loop^1_v]_{rw}$ and $\beta^1=[loop^1_h]_{rw}$. 

\end{lemma}

\begin{proof}
Consider the following cases

\item[(i)] Base case: $\beta^{0}\alpha^{0}=\rho$.
\item[(ii)] $\rho\circ \alpha=\tau(\alpha,\rho) \underset{trr}{=}\alpha =\beta^{0}\alpha^{1}.$
\item[(iii)] $\rho\circ \beta=\tau(\beta,\rho) \underset{trr}{=}\beta =\beta^{1}\alpha^{0}.$
\item[(iv)] $\rho\circ \alpha^{-1}=\tau(\sigma(\alpha),\rho) \underset{trr}{=}\sigma(\alpha) =\beta^{0}\alpha^{-1}.$
\item[(v)] $\rho\circ \beta^{-1}=\tau(\sigma(\beta),\rho) \underset{trr}{=}\sigma(\beta) =\beta^{-1}\alpha^{0}.$

Assuming, by the induction hypothesis that every path in $\mathbb{T}^2$ is \textit{rw-equal} to  $\beta^{m} \alpha^{n}$, we have:

\item[(1)] $\rho\circ \beta^{m}\alpha^{n}=\tau(\beta^{m}\alpha^{n},\rho) \underset{trr}{=}\beta^{m}\alpha^{n}.$
\item[(2)] $\alpha\circ \beta^{m}\alpha^{n}\underset{co}{=} \alpha\circ \alpha^{n}\beta^{m}=\alpha^{n+1}\beta^{m}\underset{co}{=}\beta^{m}\alpha^{n+1}.$
\item[(3)]  $\beta\circ \beta^{m}\alpha^{n}=\beta^{m+1}\alpha^{n}.$

\item[(4)] $\beta^{-1}\circ \beta^{m}\alpha^{n} =(\beta^{-1}\circ (\beta\circ \beta^{m-1}))\alpha^{n}\underset{tt}{=}((\beta^{-1}\circ \beta)\circ \beta^{m-1})\alpha^{n}\underset{tsr}{=}(\rho\circ\beta^{m-1})\alpha^{n}=\beta^{m-1}\alpha^{n}.$

\item[(4)] $\alpha^{-1}\circ \beta^{m}\alpha^{n}\underset{co}{=} \alpha^{-1}\circ \alpha^{n}\beta^{m}=(\alpha^{-1}\circ (\alpha\circ \alpha^{n-1}))\beta^{m}\underset{tt}{=} ((\alpha^{-1}\circ \alpha)\circ \alpha^{n-1})\beta^{m}\underset{tsr}{=}(\rho\circ\alpha^{n-1})\beta^{m}=\alpha^{n-1}\beta^{m} \underset{co}{=}\beta^{m}\alpha^{n-1}.$

So all paths in $\mathbb{T}^{2}$ are \textit{rw-equal} to $\beta^{m} \alpha^{n}.$
\end{proof}

\begin{proposition}
$\Pi_{1}(\mathbb{T}^2,x_0)$ provided with operations $\rho, \sigma, \tau$ is a group.
\end{proposition}
\smallskip
\begin{proof}

\item [(i)] \textbf{Closure:}  $$$$

\begin{prooftree}
\AxiomC{ $x_{0} \underset{\beta^{u}\alpha^{v}}{=}x_{0}$}
\AxiomC{$x_{0} \underset{\beta^{r}\alpha^{s}}{=} x_{0}$}
\BinaryInfC{ $x_{0} \underset{\tau \left(\beta^{u}\alpha^{v},\beta^{r}\alpha^{s}\right)}{=}x_{0}$}.
\end{prooftree}

But, 

\begin{eqnarray*}
\tau(\beta^{u}\alpha^{v},\beta^{r}\alpha^{s}) &=& (\beta^{r}\alpha^{s}) \circ (\beta^{u}\alpha^{v}) \\
 &=& \beta^{u}\alpha^{v}\beta^{r}\alpha^{s}\\
 &\underset{co}{=}& \beta^{u}\beta^{r}\alpha^{v} \alpha^{s}\\
 &=& \beta^{u+r}\alpha^{v+s }
\end{eqnarray*}

\item [(ii)] \textbf{Inverse:}$$$$ 

\begin{prooftree}
\AxiomC{ $x_{0} \underset{\beta^{m}\alpha^{n}}{=}x_{0}$}
\AxiomC{$x_{0} \underset{\sigma{(\beta^{m})}\sigma{(\alpha^{n})}}{=} x_{0}$}
\BinaryInfC{ $x_{0} \underset{\tau \left(\beta^{m}\alpha^{n},\sigma{(\beta^{m})}\sigma{(\alpha^{n})}\right)}{=}x_{0}$}.
\end{prooftree}

But, 

\begin{eqnarray*}
\tau(\beta^{m}\alpha^{n},\sigma{(\beta^{m})}\sigma{(\alpha^{n})}) &=& (\sigma{(\beta^{m})}\sigma{(\alpha^{n})}) \circ (\beta^{m}\alpha^{n}) \\
 &=& \sigma{(\beta^{m})}\sigma{(\alpha^{n})} \beta^{m}\alpha^{n}\\
 &\underset{co}{=}&  \sigma{(\beta^{m})}\beta^{m}\sigma{(\alpha^{n})}\alpha^{n}\\
 &\underset{tsr}{=}&\rho_{\beta}\rho_{\alpha}\\
 &\underset{trr}{=}&\rho_{x_{0}}.
\end{eqnarray*}

On the other hand, we have:

\begin{prooftree}
\AxiomC{ $x_{0} \underset{\sigma{(\beta^{m})}\sigma{(\alpha^{n})}}{=}x_{0}$}
\AxiomC{$x_{0} \underset{\beta^{m}\alpha^{n}}{=} x_{0}$}
\BinaryInfC{ $x_{0} \underset{\tau \left(\sigma{(\beta^{m})}\sigma{(\alpha^{n}),\beta^{m}\alpha^{n}}\right)}{=}x_{0}$}.
\end{prooftree}

But, 

\begin{eqnarray*}
\tau(\sigma{(\beta^{m})}\sigma{(\alpha^{n})},\beta^{m}\alpha^{n}) &=& (\beta^{m}\alpha^{n}) \circ (\sigma{(\beta^{m})}\sigma{(\alpha^{n})})\\
 &=& \beta^{m}\alpha^{n}\sigma{(\beta^{m})}\sigma{(\alpha^{n})}\\
 &\underset{co}{=}& \beta^{m}\sigma{(\beta^{m})}\alpha^{n}\sigma{(\alpha^{n})}\\
 &\underset{tr}{=}&\rho_{\beta}\rho_{\alpha}\underset{trr}{=}\rho_{x_{0}}.
\end{eqnarray*}


\item [(iii)] \textbf{Identity:}$$$$ 

\begin{prooftree}
\AxiomC{ $x_{0} \underset{\beta^{m}\alpha^{n}}{=}x_{0}$}
\AxiomC{$x_{0} \underset{\rho_{x_{0}}}{=} x_{0}$}
\BinaryInfC{ $x_{0} \underset{\tau \left(\beta^{m}\alpha^{n},\rho_{x_{0}}\right)}{=}x_{0}$}
\end{prooftree}

But, 

\begin{eqnarray*}
\tau(\beta^{m}\alpha^{n},\rho_{x_{0}}) &=& (\rho_{x_{0}}) \circ (\beta^{m}\alpha^{n})  \\
 &=& \rho_{x_{0}}\beta^{m}\alpha^{n}  \\
 &\underset{tlr}{=}&  \beta^{m}\alpha^{n}
\end{eqnarray*}
 
 and so \[
 \tau(\beta^{m}\alpha^{n},\rho_{x_{0}}) \underset{trr}{=} \beta^{m}\alpha^{n}.
\]

On the other hand, we have: $$$$

\begin{prooftree}
\AxiomC{ $x_{0} \underset{\rho_{x_{0}}}{=}x_{0}$}
\AxiomC{$x_{0} \underset{\beta^{m}\alpha^{n}}{=} x_{0}$}
\BinaryInfC{ $x_{0} \underset{\tau \left(\rho_{x_{0},\beta^{m}\alpha^{n}}\right)}{=}x_{0}$}.
\end{prooftree}

\bigskip
But, 

\begin{eqnarray*}
\tau \left(\rho_{x_{0}},\beta^{m}\alpha^{n}\right) &=&  (\beta^{m}\alpha^{n}) \circ (\rho_{x_{0}})  \\
 &=& \beta^{m}\alpha^{n} \rho_{x_{0}}  \\
 &\underset{trr}{=}&  \beta^{m}\alpha^{n}
\end{eqnarray*}

and so \[
 \tau(\rho_{x_{0}},\beta^{m}\alpha^{n}) \underset{trr}{=} \beta^{m}\alpha^{n}.
\]

\item [(iv)] \textbf{Associativity:}$$$$

\begin{prooftree}
\AxiomC{$x_{0} \underset{\beta^{m}\alpha^{n}}{=}x_{0}$}
\AxiomC{$x_{0} \underset{\beta^{i}\alpha^{j}}{=} x_{0}$}
\BinaryInfC{$ x_{0} \underset{\tau \left(\beta^{m}\alpha^{n},\beta^{i}\alpha^{j}\right)}{=}x_{0}$}
\AxiomC{$x_{0} \underset{\beta^{r}\alpha^{s}}{=}x_{0}$}
\BinaryInfC{$x_{0} \underset{\tau \left(\tau \left(\beta^{m}\alpha^{n},\beta^{i}\alpha^{j}\right),\beta^{r}\alpha^{s} \right)}{=}x_{0}$}.
\end{prooftree}

But, 

\begin{eqnarray*}
\tau \left(\tau \left(\beta^{m}\alpha^{n},\beta^{i}\alpha^{j}\right),\beta^{r}\alpha^{s} \right)  
 &=& (\beta^{r}\alpha^{s}) \circ \tau (\beta^{m}\alpha^{n},\beta^{i}\alpha^{j}) \\
 &=& (\beta^{r}\alpha^{s}) \circ (\beta^{i}\alpha^{j}\circ \beta^{m}\alpha^{n})\\
 &=& (\beta^{r}\alpha^{s}) \circ (\beta^{i}\alpha^{j}\beta^{m}\alpha^{n})\\
 &=& \beta^{r}\alpha^{s} \beta^{i}\alpha^{j}\beta^{m}\alpha^{n}.\\
\end{eqnarray*}
\bigskip

On the other hand, we have:

\begin{prooftree}
\AxiomC{$x_{0} \underset{\beta^{m}\alpha^{n}}{=}x_{0}$}
\AxiomC{$x_{0} \underset{\beta^{i}\alpha^{j}}{=} x_{0}$}
\AxiomC{$x_{0} \underset{\beta^{r}\alpha^{s}}{=} x_{0}$}
\BinaryInfC{$ x_{0} \underset{\tau \left(\beta^{i}\alpha^{j},\beta^{r}\alpha^{s}\right)}{=}x_{0}$}
\BinaryInfC{$x_{0} \underset{\tau (\beta^{m}\alpha^{n}, \tau \left(\beta^{i}\alpha^{j},\beta^{r}\alpha^{s})\right)}{=}x_{0}$}.
\end{prooftree}

But, 

\begin{eqnarray*}
\tau \left(\beta^{m}\alpha^{n} ,\tau \left(\beta^{i}\alpha^{j},\beta^{r}\alpha^{s}\right)\right)  
 &=&  \tau (\beta^{i}\alpha^{j},\beta^{r}\alpha^{s})\circ(\beta^{m}\alpha^{n}) \\
 &=& (\beta^{r}\alpha^{s} \circ \beta^{i}\alpha^{j})\circ (\beta^{m}\alpha^{n}) \\
 &=& (\beta^{r}\alpha^{s}\beta^{i}\alpha^{j})\circ (\beta^{m}\alpha^{n}) \\
 &=& \beta^{r}\alpha^{s} \beta^{i}\alpha^{j}\beta^{m}\alpha^{n}.\\
\end{eqnarray*}

Therefore, the structure $\Big(\Pi_{1}(\mathbb{T}^{2},x_0), \rho, \sigma, \tau \Big)$ is indeed a group. We will, for simplicity, denote it in what follows for $\Pi_{1}(\mathbb{T}^{2},x_0)$, and we will call it The Fundamental Group of $\mathbb{T}^{2}$.

\end{proof}

\begin{theorem}
$\Pi_{1}\left(\mathbb{T}^{2},x_{0}\right) \simeq \mathbb{Z} \times \mathbb{Z}.$
\end{theorem}

\begin{proof}


Consider the map:

\begin{eqnarray*}
toPath^{2}: \mathbb{Z} \times \mathbb{Z} &\longrightarrow&  \Pi_{1}\left (\mathbb{T}^{2},x_{0}\right)\\
(m,n)  &\longrightarrow&  \ \beta^{m}\alpha^{n}.
\end{eqnarray*}

\begin{itemize}
    \item [i] $toPath^{2}$ is a homomorphism.
    
    Let $(m_1+m_2,n_1+n_2) \in \mathbb{Z}\times \mathbb{Z}$, then:
   \begin{eqnarray*}
         toPath^{2}(m_1+m_2,n_1+n_2)&=&\beta^{m_1+m_2}\alpha^{n_1+n_2}\\
         &=&\beta^{m_1}\beta^{m_2}\alpha^{n_1}\alpha^{n_2}\\
         &\underset{co}{=}&\beta^{m_1}\alpha^{n_1}\beta^{m_2}\alpha^{n_2}\\
         &=&toPath^2(m_1,n_1)\circ toPath^2(m_2,n_2).
  \end{eqnarray*}

  \item [ii] $toPath$ is surjective.

By \textbf{Lemma \ref{pathloopMN}}, as every path in $\mathbb{T}^2$ is $rw$-equal to a path $\beta^{m}\alpha^{n}$, we have that for all paths $ \beta^{i}\alpha^{j} \in \Pi_{1}(\mathbb{T}^2), \exists (i,j) \in \mathbb{Z}\times \mathbb{Z}$, such that, $ toPath(i,j)=\beta^{i}\alpha^{j}$.

   \item [iii] $Ker(toPath^{2})=\{(0,0)\}$.
   
   Suppose that $(m,n)\neq (0,0) \in \mathbb{Z}\times\mathbb{Z}$, such that $(m,n) \in Ker(toPath^2)$. Thus,
   
   \begin{eqnarray*}
      toPath^2(m,n)&=&toPath^2(m+0,n+0)\\
     &\overset{hom}{=}&toPath^2(m,n)\circ toPath^2(0,0)\\ 
     &=&\beta^{m}\alpha^{n}\beta^{0}\alpha^{0}\\
     &=&\tau(\beta^{0}\alpha^{0},\beta^{m}\alpha^{n})\\
     &=&\tau(\rho,\beta^{m}\alpha^{n})\\
     &=&\rho.
   \end{eqnarray*}
   If $\tau(\rho,\alpha)=\rho$, by $rw$-rule $\underset{tr}{\triangleright}$ we have, $\alpha= \sigma(\rho) \Rightarrow (m,n)=(0,0)  \rightarrow \leftarrow$. Therefore,  $$Ker(toPath^2)=\{(0,0)\}.$$
\end{itemize}

As $ToPath^2$ is a homomorphism surjective with $Ker(toPath^2)=\{(0,0)\}$, then $toPath^2$ is an isomorphism, that is, $\Pi_{1}(\mathbb{T}^2) \simeq \mathbb{Z}\times\mathbb{Z
}$.

\end{proof}

\subsection{Fundamental group of the real projective plane}

The real projective plane, denoted by $\mathbb{RP}^{2}$, is by definition the set of all straight lines that pass through the origin of space $\mathbb{R}^3$. We can define each of these lines by a position vector $v_r$, with $ \left \| v_r \right \|\neq 0$, so we have that $\mathbb{RP}^{2}$ is a quotient space of $\mathbb{R}^{3}-{(0,0)}$ under the equivalence relation $v_r \sim \lambda v_r$ for scalars $\lambda \neq 0$. If we impose the condition that the vectors $ \left \| v_r \right \|=1$ then $\mathbb{RP}^{2}$ is a quotient space $\mathbb{S}^{2}$ under the equivalence relation $v_r \sim -v_r$, the sphere
with antipodal points identified, where $v_r$ is position vector.

Let $[v_r]=[x, y, z]$, where $[x, y, z]=\{v_r=(x,y,z),-v_r=(-x,-y,-z)\}$ with $z\neq 0$. This is equivalent to saying that $\mathbb{RP}^{2}$ is the quotient
space of an upper hemisphere $\mathbb{D}^{2}$ with antipodal points of $\partial\mathbb{D}^{2}$ identified, as shown in \textbf{Figure \ref{fig3}}.

\begin{figure}[!htb]
\centering
\includegraphics[width=0.25\columnwidth]{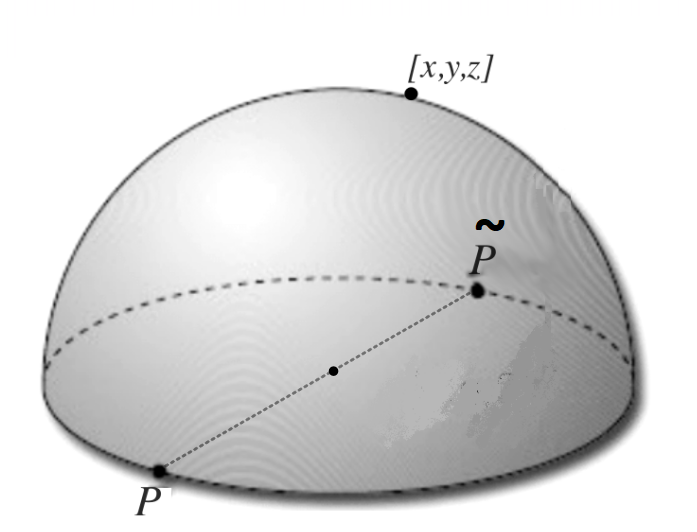}
\caption{ $P$ and $\tilde{P}$ are antipodal points in $\partial\mathbb{D}^{2}$.} 
\label{fig3}
\end{figure}

Let us then map it on the unit disk through the following map $[x,y,z] \longrightarrow (x,y,0)$, as follows in \textbf{Figure \ref{fig4}}.

\begin{figure}[H]
\centering
\includegraphics[width=0.5 \columnwidth]{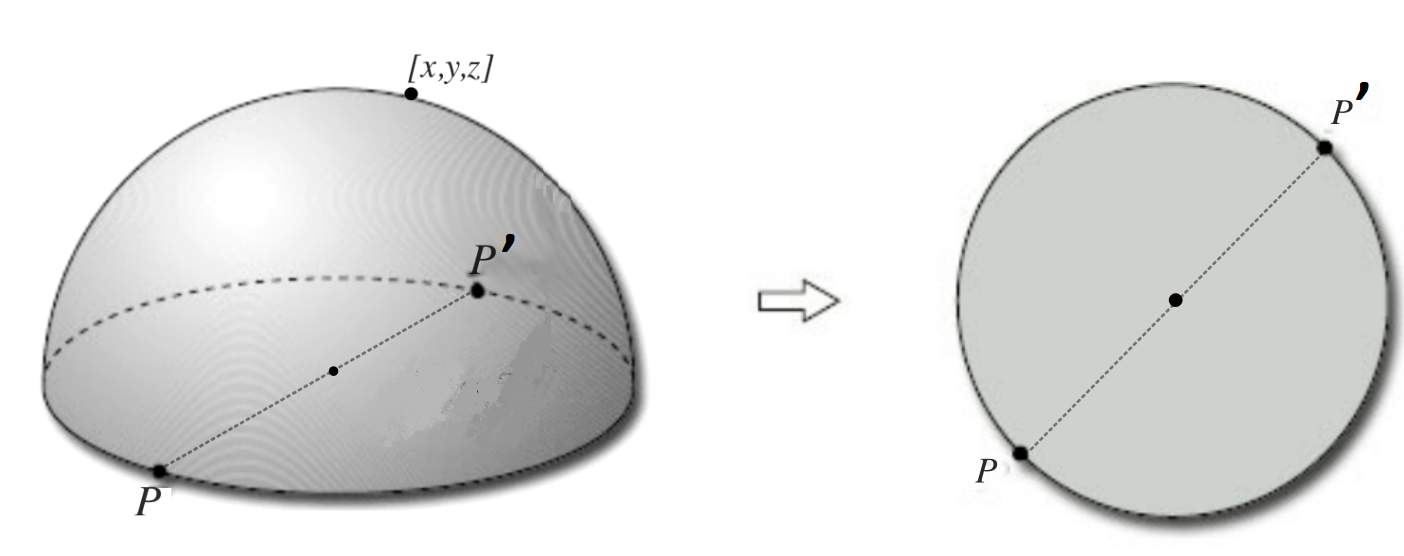}
\caption{Mapping Projection on the unit disk on the $xy$ plane.} 
\label{fig4}
\end{figure}

This way we have that $\mathbb{RP}^{2}$ is a quotient space of $D$ with antipodal points of $\partial D$ identified. Therefore we can study the fundamental group of $\mathbb{RP}^{2}$ from the disk shown on the right side of \textbf{Figure \ref{fig4}}. 

We denote by $\alpha$ any \textit{loop} that connects the identified antipodal points, so we can consider $\alpha$ as a \textit {loop}  (as follows in \textbf{Figure \ref{fig5}}) and any other \textit {loop} that connects the identified antipodal points is homotopic to $\alpha$. Note that $\forall Q \in \mathbb{D}$, any loop based on $Q$ is homotopic to the point, and it is not in our interest to study those.

\begin{figure}[H]
\centering
\includegraphics[width=0.25\columnwidth]{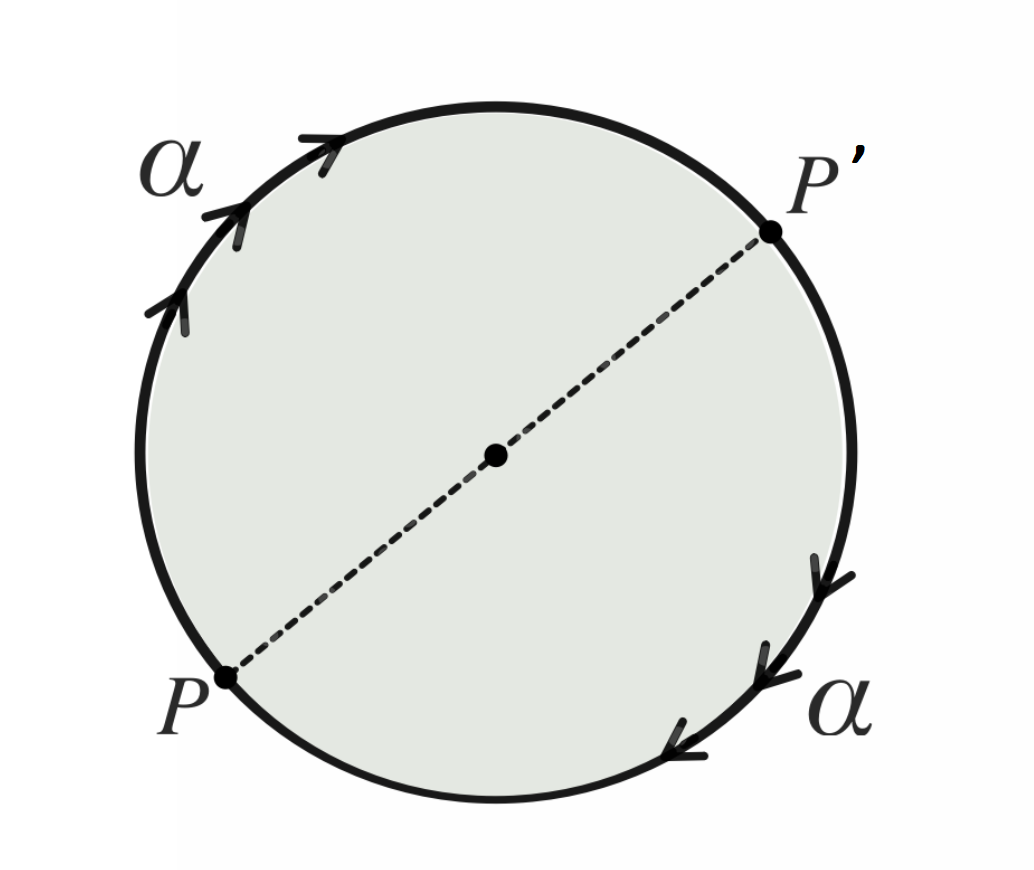}
\caption{\textit{loop} $\alpha$. } 
\label{fig5}
\end{figure}

Since we can represent the real projective plane $\mathbb{RP}^2$ as a disk $\mathbb{D}$, we can define $\mathbb{RP}^2$, homotopically, as follows: 

\begin{definition}
   
 The real projective plane $\mathbb{RP}^2$ is defined by:
 
 \item[(i)] The types $Q:\mathbb{D}$, such that $Q\in \mathbb{D}.$
 \item[(ii)] The pair $P,P':\partial D$, such that: $P,P'$ are the pairs of antipodal points identified in $\partial \mathbb{D}$.  
 \item[(iii)] A path $\alpha$ such that:   $P\underset{\alpha}{=}P'$.
 \item[(iv)] A path $cicl$ that establishes $\alpha \circ \alpha \underset{cicl}{=} \rho$, i.e,  $cicl:Id_{\mathbb{P}^2}(\alpha\circ \alpha,\rho)$.
 
\end{definition}

\begin{lemma}
\label{lemma3.5} 
All paths in $\mathbb{RP}^{2}$ can be expressed in terms of $\rho$,$\tau$, $\sigma$ and their applications, starting from the base paths $\rho$ or $\alpha$.
\end{lemma}

\begin{proof}

Consider the following base cases:

\item [\textbf{Base case   $\varphi=\rho$:}]
\item[(i)] Trivial case.

\item [\textbf{Base case $\varphi=\sigma(\phi)$:}]

\item[(i)] $\varphi=\sigma(\alpha)\underset{cicl}{=}\alpha$.

\item[(ii)] $\varphi=\sigma(\rho)\underset{rw}{=}\rho$.

\item [\textbf{Base case   $\varphi=\tau(\phi,\kappa)$:}]

\item[(i)] $\varphi=\tau(\rho,\rho) \underset{trr}{=}\rho$
\item[(ii)] $\varphi=\tau(\alpha,\rho) \underset{trr}{=}\alpha$
\item[(iii)] $\varphi=\tau(\rho,\alpha) \underset{tlr}{=}\alpha$
\item[(iv)] $\varphi=\tau(\alpha,\alpha) \underset{cicl}{=}\rho$

\item \textbf{Inductive case:} Assuming true for $n$, we have: 

\item If  \textit{$[loop^{n}]_{rw}=[\rho]_{rw}$}, we have two possibilities for $n+1$:
\item[(i)] $[loop^{n+1}]_{rw}=[loop^{n}]_{rw}\circ[\rho]_{rw}  =\tau(\rho,\rho)  \underset{trr}{=} [\rho]_{rw}$.
\item[(ii)] $[loop^{n+1}]_{rw}=[loop^{n}]_{rw}\circ \alpha = [\rho]_{rw} \circ \alpha =\tau(\alpha,\rho)  \underset{trr}{=} \alpha$.

\item If  \textit{$[loop^{n}]_{rw}=\alpha$}, we have two possibilities for $n+1$:

\item[(i)] $[loop^{n+1}]_{rw}=[loop^{n}]_{rw}\circ\alpha = \alpha \circ \rho =\tau(\rho,\alpha)  \underset{tlr}{=} \alpha$.
\item[(ii)] $[loop^{n+1}]_{rw}=[loop^{n}]_{}\circ \alpha = \alpha \circ \alpha =\tau(\alpha,\alpha)  \underset{cicl}{=} [\rho]_{rw}$.

\end{proof}

Thus, all paths in $\mathbb{RP}^{2}$ generated by $\rho$ or $\alpha$ are \textit{rw-equal} to either $\alpha$ or $\rho$. Since we have $ \alpha \circ \alpha =\tau(\alpha,\alpha)  \underset{cicl}{=} [\rho]_{rw}$, the term $cicl$ give us one important result: $\alpha=\sigma(\alpha)$.

\begin{proposition}
$\Pi_{1}(\mathbb{RP}^{2})$ provided with operations $\rho, \sigma, \tau$ is a group.
\end{proposition}

\begin{proof}

\item [(i)] \textbf{Closure:}$$$$

\begin{prooftree}
\AxiomC{ $P \underset{\alpha}{=}P$}
\AxiomC{$P \underset{\alpha}{=}P$}
\BinaryInfC { $P \underset{\tau \left(\alpha,\alpha\right)}{=}P$}
\end{prooftree}

But,  $$\alpha \circ \alpha = \tau \left(\alpha,\alpha\right) \underset{cicl}{=} \rho \in \Pi_{1}\left(\mathbb{RP}^2\right).$$


\item[(ii): \textbf{Inverse:}]$$$$

\begin{prooftree}
\AxiomC{ $P \underset{\alpha}{=} P$}
\AxiomC{$P \underset{\sigma{(\alpha)}}{=} P$}
\BinaryInfC{ $P\underset{\tau \left(\alpha,\sigma{(\alpha)}\right)}{=}P$}
\end{prooftree}

But,  

$$\sigma(\alpha) \circ \alpha= \tau \left(\alpha,\sigma(\alpha)\right) \underset{tr}{=} \rho \in \Pi_{1}\left(\mathbb{RP}^2\right).$$

On the other hand, we have:

\begin{prooftree}
\AxiomC{$P \underset{\sigma{(\alpha)}}{=} P$}
\AxiomC{ $P \underset{\alpha}{=} P$}
\BinaryInfC{ $P\underset{\tau \left(\sigma{(\alpha)}, \alpha \right)}{=}P$}
\end{prooftree}

But, 

 $$ \alpha \circ \sigma(\alpha) = \tau \left(\sigma(\alpha), \alpha\right) \underset{tsr}{=} \rho \in \Pi_{1}\left(\mathbb{RP}^2\right).$$


\item[(iii)] \textbf{Identity:}$$$$

\begin{prooftree}
\AxiomC{ $P \underset{\alpha}{=} P$}
\AxiomC{$P \underset{\rho}{=} P$}
\BinaryInfC{ $P \underset{\tau \left(\alpha,\rho\right)}{=} P$}
\end{prooftree}

But, 

 $$\rho \circ \alpha = \tau \left(\alpha, \rho\right) \underset{tlr}{=} \alpha \in \Pi_{1}\left(\mathbb{RP}^2\right).$$

On the other hand, we have:

\begin{prooftree}
\AxiomC{$P \underset{\rho}{=} P$}
\AxiomC{ $P \underset{\alpha}{=} P$}
\BinaryInfC{ $P \underset{\tau \left(\rho, \alpha\right)}{=} P$}
\end{prooftree}

But, 

$$\alpha \circ \rho = \tau \left(\rho,\alpha\right) \underset{trr}{=} \alpha \in \Pi_{1}\left(\mathbb{RP}^2\right).$$


\item[(iv )] \textbf{Associativity:}]$$$$

\begin{prooftree}
\AxiomC{ $P \underset{\alpha}{=} P$}
\AxiomC{ $P \underset{\alpha}{=} P$}
\BinaryInfC{$ P \underset{\tau \left(\alpha,\alpha\right)}{=}P$}
\AxiomC{$P \underset{\alpha}{=}P$}
\BinaryInfC{$P \underset{\tau \left(\tau \left(\alpha,\alpha\right),\alpha \right)}{=}P$}
\end{prooftree}

But, 

\begin{eqnarray*}
\tau \left(\tau \left(\alpha,\alpha\right),\alpha \right)  
 &=& \alpha \circ \tau \left(\alpha,\alpha\right) \\
 &\underset{cicl}{=}& \alpha \circ \rho\\
 &=& \tau \left(\rho,\alpha\right)\\
 &\underset{trr}{=}& \alpha\\
\end{eqnarray*}
\bigskip

On the other hand, we have:

\begin{prooftree}
\AxiomC{$P \underset{\alpha}{=}P$}
\AxiomC{$P \underset{\alpha}{=} P$}
\AxiomC{$P \underset{\alpha}{=} P$}
\BinaryInfC{$ P \underset{\tau \left(\alpha,\alpha\right)}{=}P$}
\BinaryInfC{$P \underset{\tau (\alpha, \tau \left(\alpha,\alpha)\right)}{=}P$}
\end{prooftree}

But, 

\begin{eqnarray*}
\tau (\alpha, \tau \left(\alpha,\alpha)\right)  
 &=&  \tau (\alpha,\alpha)\circ\alpha \\
 &\underset{cicl}{=}& \rho \circ \alpha \\
 &=& \tau \left(\alpha, \rho\right) \\
&\underset{tlr}{=}& \alpha\\
\end{eqnarray*}

Since  $\tau \left(\tau \left(\alpha,\alpha\right),\alpha \right)= \tau (\alpha, \tau \left(\alpha,\alpha)\right)$, it follows that associativity is valid and therefore $\left( \Pi_{1}(\mathbb{RP}^{2}),\circ \right)$ is a group generated by $\rho$ and $\alpha$.

\end{proof}

\begin{theorem}
$ \Pi_{1}(\mathbb{RP}^{2}) \simeq \mathbb{Z}_{2}$.
\end{theorem}

\begin{proof}

Consider the application defined and denoted by: 

\begin{eqnarray*}
toPath_{\mathbb{Z}_2}: &&\mathbb{Z}_2 \rightarrow \Pi_{1}\left (\mathbb{RP}^2\right)\\
&&z \rightarrow  toPath_{\mathbb{Z}_2}= [loop^{z}]_{rw}.
\end{eqnarray*}

\begin{itemize}
    \item [(i)] $toPath_{\mathbb{Z}_2}$ is a homomorphism.

 Let $z_1$ and $z_2 \in \mathbb{Z}_{2}$, then:
   \begin{eqnarray*}
         toPathh_{\mathbb{Z}_2}(z_1 +z_2)&=&[loop^{z_1+z_2}]_{rw}\\
         &=&\tau([loop^{z_1}]_{rw},[loop^{z_2}]_{rw})\\
         &=&toPathh_{\mathbb{Z}_2}(z_2)\circ toPathh_{\mathbb{Z}_2}(z_1).
  \end{eqnarray*}  
  
  On the other hand, we have:
  
  \begin{eqnarray*}
         toPathh_{\mathbb{Z}_2}(z_2 +z_1)&=&[loop^{z_2+z_1}]_{rw}\\
         &=&\tau([loop^{z_2}]_{rw},[loop^{z_1}]_{rw})\\
         &=&toPathh_{\mathbb{Z}_2}(z_1)\circ toPathh_{\mathbb{Z}_2}(z_2).
  \end{eqnarray*}  
  
  Thus, $toPath_{\mathbb{Z}_2}(z_1+z_2)=toPath_{\mathbb{Z}_2}(z_1)\circ toPath(z_2)$.
  
  \item [(ii)] $toPath_{\mathbb{Z}_2}$ is surjective.

By \textbf{Lemma \ref{lemma3.5}}, every path in $\mathbb{RP}^{2}$ is $rw$-equal to $\rho$ and $\alpha$. So given any path in $\Pi_{1}(\mathbb{RP}^{2})$, for $z=0$ and $z=1$ we have $\rho=toPath_{\mathbb{Z}_2}(0)$ and $\alpha=toPath_{\mathbb{Z}_2}(1)$, respectively.

   \item [(iii)] $Ker(toPath_{\mathbb{Z}_2})=\{0\}$.
   
  By \textbf{Lemma \ref{lemma3.5}}, there is only one element in $z\in \mathbb{Z}_{2}$ such that $toPathh_{\mathbb{Z}_2}(z)=0$. Therefore,   $Ker(toPath_{\mathbb{Z}_2})=\{0\}.$

\end{itemize}

$$toInt = \begin{cases}
toInt([loop^{0}]_{rw}=[\rho]_{rw})=0&  \\
toInt([loop^{1}]_{rw}=\alpha)=1&\\
\end{cases}$$
$$$$

Thus, the isomorphism holds.
\end{proof}
	
\section{Conclusion}
	
    Our main objective has been the calculation of the fundamental groups of many surfaces using  a labelled deduction system based on the concept of computational paths (sequences of rewrites).
    The main advantage of this approach is that we avoid the use of more complex techniques, such as those made in algebraic topology in pure mathematics or by the method of encoding-decoding used in homotopy type theory.  As a consequence, our calculations proved to be straightforward and simple. Using computational paths as our main tool, we have calculated the fundamental group of the circle, torus and projective plane. Therefore, we have shown that it is possible to use the theory of computational paths to obtain useful results in algebraic topology.
    
    Finally, an almost natural question for our study would be: is it possible to calculate the fundamental group of the Klein bottle using the same technique? This question is going to be a guiding element to develop our future studies.

\bibliographystyle{unsrt}  

\end{document}